\documentclass[12pt]{article}
\usepackage[utf8]{inputenc}
\usepackage{amsmath,amssymb,amsthm,algorithmicx,algorithm,setspace,url, color}
\usepackage{graphicx,bbm}
\usepackage{natbib, minitoc}
\usepackage{authblk,booktabs,epstopdf,color,lscape,float, subcaption, caption}
\usepackage[hidelinks]{hyperref}
\hypersetup{
    colorlinks,
    linkcolor={blue},
    citecolor={blue},
    urlcolor={blue}
}
\definecolor{red}{rgb}{0.45,0.0,0.0}
\usepackage{caption,subcaption, comment}
\usepackage{multirow}
\usepackage[space]{grffile}
\date{}
\def\Gauss{{\mathrm{N}}}

\def\T{\mathrm{\scriptscriptstyle{T}}}

\newtheorem{theorem}{Theorem}[section]
\newtheorem{lemma}[theorem]{Lemma}

\newtheorem{proposition}[theorem]{Proposition}

\newtheorem{rem}{Remark}[section]
\newcommand{\norm}[1]{\left\lVert #1 \right\rVert}
\DeclareMathOperator*{\argmin}{arg\,min}
\usepackage[finalnew]{trackchanges}
\addeditor{ab}
\addeditor{ac}

\def\cred{\color{black}}
\soulregister\cite7
\soulregister\citep7
\soulregister\ref7
\newcommand{\blind}{1}

\begin{document}

\def\spacingset#1{\renewcommand{\baselinestretch}%
{#1}\small\normalsize} \spacingset{1}


\if1\blind
{
  \title{\bf Orthogonal calibration via posterior projections with applications to the Schwarzschild model}
  \author{Antik Chakraborty\\
    Department of Statistics, Purdue University\\
    Jonelle L. Walsh \\
    George P.~and Cynthia Woods Mitchell Institute for Fundamental Physics and Astronomy, Department of Physics and Astronomy, Texas A\&M University\\
    Louis Strigari \\
    George P.~and Cynthia Woods Mitchell Institute for Fundamental Physics and Astronomy, Department of Physics and Astronomy, Texas A\&M University\\
    Bani K. Mallick\\
    Department of Statistics, Texas A\&M University\\
    Anirban Bhattacharya \\
    Department of Statistics, Texas A\&M University\\
    
    }
    
  \maketitle
} \fi

\if0\blind
{
  \bigskip
  \bigskip
  \bigskip
  \begin{center}
    {\LARGE\bf \bf Orthogonal calibration via posterior projections with applications to the Schwarzschild model}
\end{center}
  \medskip
} \fi
\bigskip
\begin{abstract}
The orbital superposition method originally developed by \cite{schwarzschild1979numerical} is used to study the dynamics of growth of a black hole and its host galaxy, and has uncovered new relationships between the galaxy's global characteristics. Scientists are specifically interested in finding optimal parameter choices for this model that best match physical measurements along with quantifying the uncertainty of such procedures. This renders a statistical calibration problem with multivariate outcomes. In this article, we develop a Bayesian method for calibration with \emph{multivariate outcomes} using orthogonal bias functions thus ensuring parameter identifiability. Our approach is based on projecting the posterior to an appropriate space which allows the user to choose any nonparametric prior on the bias function(s) instead of having to model it (them) with Gaussian processes. We develop a functional projection approach using the theory of Hilbert spaces. A finite-dimensional analogue of the projection problem is also considered. We illustrate the proposed approach using a BART prior and apply it to calibrate the Schwarzschild model illustrating how a multivariate approach may resolve discrepancies resulting from a univariate calibration. 
\end{abstract}

\noindent%
{\it Keywords:}  Bayesian; Computer model; Emulator; Scientific modeling; Regression trees.
\vfill

\newpage
\spacingset{1.95} 
\section{Introduction}
\label{sec:intro}
We focus on the problem of computer code calibration motivated by an application in astrophysics that involves the orbital superposition method due to \cite{schwarzschild1979numerical}. These orbit-based models are the main method for measuring the masses of supermassive black holes at the centers of nearby galaxies \citep{kormendyho2013}, and have led to the establishment of empirical relationships between the masses of black holes and large-scale galaxy properties, which surprisingly suggest that black holes and their host galaxies grow and evolve together over time \citep{mcconnell2013, kormendyho2013, saglia2016}. 

In their seminal work \cite{kennedy2001bayesian}, the authors proposed a framework that simultaneously models the \change[ab]{observed data of}{data observed from} a physical process and data generated from computer code implementations of mathematical models \change[ab]{of}{emulating} the physical process. The fundamental contribution of their work was to include a discrepancy/bias function in their model, which the authors interpreted as an unknown function that captures the inability of the mathematical model to explain variations in the observed data. A Bayesian nonparametric regression approach based on Gaussian processes \citep{Rasmussen:2005:GPM:1162254} was then used to model the unknown bias function as well as the computer code \change[ab]{that simulates}{simulating} the physical process. This enabled fast {\it emulation} of the computer code simulators and produced rapid estimates of input parameters of the computer code (called calibration parameters) that best approximated the observed data. Since then several authors have proposed statistical methods for computer code calibration; see for example \cite{higdon2008computer, bayarri2007computer, chakraborty2013spline, pratola2016bayesian} among others.


\cite{tuo2015efficient, tuo2016theoretical} studied the general calibration model of \cite{kennedy2001bayesian} from a theoretical point-of-view and showed that the calibration parameters are not identifiable. \note[ab]{\cred ** Is this all they did? We should add if they suggested a remedy, like the next line. Also, this was one giant line; I have added some stops to break it into multiple lines. -- AB(1/19).} \add[ac]{As a remedy, \cite{tuo2015efficient} } defined the calibration parameter as the point in the input parameter space that minimizes the $L^2$ loss function between the physical process and the computer code. \add[ac]{\cite{plumlee2017bayesian} extended it to more general losses} and showed under certain conditions, imposing an orthogonality constraint on the bias function reduces uncertainty in the calibration parameter. The author also proposed a modified covariance function for Gaussian process priors used for the bias function. \cite{gu2018scaled} proposed a scaled-Gaussian process in which the authors constrain the ``size" of the bias function; \cite{xie2021bayesian} interpreted the calibration parameters as functional of the bias function modeled using a Gaussian process. However, as indicated by other authors \citep{pratola2016bayesian} Gaussian process priors often suffer from scalability issues to big data sets and high-dimensional input spaces. Moreover, it is not immediate how to modify the orthogonality condition in the presence of multiple outputs of the physical process and the computer code. 

In this article, we propose a generalized Bayesian approach for computer code calibration with multivariate output which takes into account identifiability issues pointed out by \cite{tuo2016theoretical}. We first generalize the restrictions on the bias function to the situation where there are multiple potentially correlated outcomes compared to the univariate case addressed in \cite{plumlee2017bayesian}. We then take a posterior projection approach wherein one is free to specify any nonparametric prior on the bias function \emph{apriori}, and then the posterior is projected onto the relevant space. This allows more flexible nonparametric prior specification such as BART \citep{chipman2010bart}, and the projection automatically takes care of the constraints built into the model. \add[ac]{For inference based on the projected posterior we develop two algorithms for posterior sampling. The first one is essentially a Hilbert projection on the space of square-integrable functions and the second one borrows inspiration from a projection interpretation of the multivariate Gaussian distribution. Numerical experiments in Section \ref{sec:simulations} show superiority of the proposed method both in terms of computational scalability and modeling flexibility.}

\note[ab]{\cred ** Paragraph change. -- AB(1/19)} 

\note[ac]{** commented out the paragraph here and moved its contents to the previous paragraph}

The rest of the paper is organized as follows. In Section \ref{sec:preliminary}, we review the basic concepts of orthogonal calibration. Next, in Section \ref{sec:posterior_projection}, we introduce our proposed posterior projection method. Section \ref{sec:unknown_computer_model} discusses extensions to cases when the computer model is not available explicitly. Numerical experiments comparing the proposed method to existing approaches are given in Section \ref{sec:simulations}. We end with an application to calibrating the Schwarzschild's model in Section \ref{sec:schwarzchild_model}.


\section{Orthogonal calibration of computer models}\label{sec:preliminary}
\add[ab]{We begin with a quick review of orthogonal calibration for the case when the observed process is univariate.} Suppose we are interested in a real physical system $y_R(\cdot)\in \mathbb{R}$ with controllable inputs $x \in \mathcal{X}$ where $\mathcal{X}$ is some \change[ab]{open}{closed}, bounded subset of $\mathbb{R}^p$. We do not directly observe the system. Instead, we observe stochastic field observations $y_F(\cdot)$ \add[ab]{at $n$ input points called the design $\mathcal{D}=\{x_1, \ldots, x_n\}$,} for which \change[ab]{we assume the following model}{a natural stochastic model is}
    $y_F(x_i) = y_R(x_i) + \epsilon(x_i), \, \epsilon(x_i) \overset{ind.} \sim \Gauss(0, \sigma^2_F), \, i = 1, \ldots, n.$
\remove[ab]{Consider available observations on $n$ input points called the design $\mathcal{D}=\{x_1, \ldots, x_n\}$.} \change[ab]{In}{Suppose in} addition to the field observations, we also have a simulator \add[ab]{$f: \mathcal{X} \times \Theta \to \mathbb{R}$}  for the physical system \add[ab]{(depending on the controllable input $x$ as well as other parameters $t \in \Theta \subset \mathbb{R}^p$)} which typically comes in the form of \add[ab]{a} computer code/model. For now, we assume that the functional form of the computer model $f(x, t)$ is explicitly known for every $x \in \mathcal{X}, t \in \Theta$ and defer the extension of the proposed method to unavailable $f$ later. \remove[ab]{The simulator depends on the controllable inputs along with other parameters $t \in \Theta \subset \mathbb{R}^p$. Let $f(x, t) \in \mathbb{R}^q$ denote the computer model.} The parameter $t$ represents our understanding of the physical \change[ab]{system in terms of several of its attributes. We hope}{system, and we hope} that for some unknown \change[ab]{$\theta \in \Theta$}{$\theta^* \in \Theta$}, the computer model \add[ab]{closely} approximates the physical process $y_R(\cdot)$. \remove[ab]{really well.}Formally, suppose $L\{y_R(\cdot), f(\cdot, \cdot)\}$ is a loss function that can distinguish between the best possible parameter \change[ab]{$\theta$}{$\theta^*$} and all other parameter values $t$. {\color{black} One such example of a loss function is the squared error loss: $L\{y_R(\cdot), f(\cdot, t)\} = \int_{\mathcal{X}} (y_R(x) - f(x, t))^2 dx$.}
If there exists a unique $t^*$ such that $y_R(x) = f(x, t^*)$ for all $x \in \mathcal{X}$ and the loss $L$ is {\em strictly consistent} following \cite{gneiting2011making}, then \change[ab]{$\theta = t^*$}{$\theta^* = t^*$}. However, existence of such a $t^*$ is not always guaranteed for most practical computer models and loss functions. In the absence of such a parameter, the computer model is biased and \change[ab]{$\theta$}{$\theta^*$} is defined as the parameter combination at which the loss $L\{y_R(\cdot), f(\cdot, \cdot)\}$ is minimized, \add[ac]{ i.e. $\theta^* = \argmin_{t \in \Theta}L\{y_R(\cdot), f(\cdot, t)\}$}. \add[ab]{To account for such bias,} \cite{kennedy2001bayesian} \remove[ab]{in a seminal work,} proposed the following model for the observed data on the system 
\begin{equation}\label{eq:kennedy_ohagan}
    y_F(x_i) = f(x_i, \theta) + b_{\theta}(x_i) + \epsilon(x_i), \, \epsilon(x_i) \overset{ind.} \sim \Gauss(0, \sigma^2_F), \, i = 1, \ldots, n,
\end{equation}
where $\theta$ \change[ab]{is the unknown parameter that minimizes the loss discussed above}{represents the model parameter that targets the population parameter $\theta^*$ defined above}, and $b_{\theta}(x)$ is interpreted as the discrepancy or bias between the physical process and the assumed computer model, i.e. $b_{\theta}(x) = y_R(x) - f(x, \theta) $. \cite{kennedy2001bayesian}  considered Gaussian process priors \citep{Rasmussen:2005:GPM:1162254} as a prior distribution over the unknown function $b_{\theta}(x)$ together with a Normal/Uniform prior over the parameter $\theta$.  \remove[ac]{Refer to \cite{plumlee2017bayesian} for a more detailed definition of the parameter $\theta$ and the bias function $b_{\theta}(\cdot)$.} \note[ab]{\cred ** Do we need the previous line? I suggest deleting it but wanted to double check. -- AB(1/20)} Suppose the prior mean function of $b_\theta$ is 0 and the covariance function is $C(\cdot, \cdot)$. Letting $\pi(\theta)$ to be the prior over $\theta$, with $n$ observations on $y_F(\cdot)$ on the design $\mathcal{D}$, the posterior distribution of $(\theta, \mathbf{b}), \, \mathbf{b} = \{b_{\theta}(x_1), \ldots, b_{\theta}(x_n)\}^\T$ can be obtained using Bayes theorem as
    $\Pi(\theta, \mathbf{b}\mid y^{(n)}) \propto \ell(y^{(n)}\mid \theta, \mathbf{b}) \Pi(\theta) \Pi(\mathbf{b})$,
where given $\theta$, $y^{(n)} = \{y_F(x_1) - f(x_1, \theta), \ldots, y_F(x_n) - f(x_n, \theta)\}^\T$ and $\ell(y^{(n)}\mid \theta, \mathbf{b})$ is the likelihood corresponding to a Gaussian distribution with mean $\mathbf{b}$ and covariance matrix $\sigma^2_F \mathrm{I}_n$. The marginal posterior distribution of $\theta$, $\Pi(\theta \mid y^{(n)})$ can be obtained by observing that $y^{(n)} \mid \theta \sim \Gauss(0, K+ \sigma_F^2\mathrm{I}_n)$ with $K_{ij} = C(x_i, x_j), \, i, j = 1, \ldots, n$. 

\cite{tuo2016theoretical} proved that vanilla Gaussian process prior over the bias $b_{\theta}(\cdot)$ may lead to inaccurate inference on \change[ab]{$\theta$}{$\theta^*$}, defined as the minimizer of the loss function discussed above, even when observations are available for a large number of points $n$. This is due to the fact that without added restrictions on $b_\theta$, $\theta$ is not identifiable in \eqref{eq:kennedy_ohagan}. 

{\color{black} In light of this, \cite{plumlee2017bayesian} argued to impose a restriction on the bias function induced by the stationarity of the loss at $\theta$.
Let $g(x, \theta) = \frac{\partial f(x, \theta)}{\partial \theta}$ denote the partial derivative (assuming its existence) and suppose the loss under consideration is the squared error loss. Let us also assume regularity conditions that allow interchanging differentiation and integration. The condition given by \cite{plumlee2017bayesian} is
\begin{equation}\label{eq:plumee_orthogonality}
    \int_\mathcal{X} g(x, \theta) b_\theta(x) dx = 0.
\end{equation}
In other words, the bias function must be orthogonal to the first derivative of the computer model with respect to $\theta$. To ensure this condition is satisfied by $b_\theta$, \cite{plumlee2017bayesian} suggested a modified covariance kernel for $b_{\theta}(\cdot)$ such that realizations of such a Gaussian process prior satisfy \eqref{eq:plumee_orthogonality} almost surely. This resolution, while extremely beneficial, raises two questions.
First, how should one extend the orthogonality condition in the presence of more than one outcome. Second, and more importantly, is there an automatic way to incorporate this constraint while working with any nonparametric prior on the bias function. In the next section, we first extend the orthogonality condition to multivariate outcomes and then introduce the method of posterior projection, which is agnostic to the choice of the prior distribution on $b_\theta$.
}

\section{Orthogonal calibration via posterior projection} \label{sec:posterior_projection}
{\color{black} Consider the general case of multivariate outcomes of the process $y_R(\cdot) \in \mathbb{R}^q$. The computer model $f: \mathcal{X} \times \Theta \to \mathbb{R}^q$, where $q \geq 1$. The multivariate version of \eqref{eq:kennedy_ohagan} is 
\begin{equation}\label{eq:multivariate_kennedy_ohagan}
    y_F(x_i) = f(x_i, \theta) + b_{\theta}(x_i) + \epsilon(x_i), \, \epsilon(x_i) \overset{ind.} \sim \Gauss_q(0, \Sigma_F), \, i = 1, \ldots, n,
\end{equation}
where $\Sigma_F^{q \times q}$ is the error covariance matrix and $b_\theta(x) = (b_{\theta, 1}(x), \ldots, b_{\theta, q}(x))$. Here, $b_{\theta, k}$ is the bias of the computer model for the $k$-th process $y_{R, k}(\cdot)$, $k = 1, \ldots, q$.   
}
As mentioned earlier, the modified covariance kernel by \cite{plumlee2017bayesian} seemingly resolves idenitifiability issue of $\theta$ for $q=1$, but it hinges on the assumption that the user chooses to model $b_{\theta}(x)$ with a Gaussian process. This leaves out a plethora of other nonparametric priors for functions that have been widely popular in the literature, for example Bayesian additive regression trees (BART) \citep{chipman2010bart}, Bayesian multivariate adaptive regression splines \citep{denison1998bayesian}, \add[ac]{Bayesian neural networks \citep{neal2012bayesian}} \note[ab]{\cred ** add a Bayes neural net cite?} etc. 
Furthermore, these methods have been successfully applied in computer code calibration problems \citep{pratola2016bayesian, higdon2008computer, chakraborty2013spline}. Many of these priors have been shown to have optimal theoretical properties along with sufficient computational scalability \citep{rovckova2020posterior, linero2018bayesian}. In this article, our aim is to develop a procedure where a user is free to choose any nonparametric prior for $b_{\theta}(x)$ and the procedure automatically takes care of orthogonality constraints discussed above. 

Let $L^2_q(\mathcal{X})$ be the space of $q$-dimensional square-integrable functions on $\mathcal{X}$ equipped with the inner-product $\langle f_1, f_2 \rangle = \sum_{k=1}^q \int f_{1,k}(x)f_{2,k}(x) dx $.  Then $L^2_q(\mathcal{X})$ is a tensor Hilbert space with norm $\norm{\cdot}_{L^2_q}$ induced by the inner product: $\norm{f}_{L^2_q} = \langle f, f\rangle = \sum_{k=1}^q \int f_{k}^2(x) dx $. Define for every $t \in \Theta$,  $g_{j,k}(x, t) = \frac{\partial}{\partial t_j} f_{j,k}(x, t)$, for $j =1, \ldots, p$ and $k = 1, \ldots, q$.
We assume that for every $t \in \Theta$,  $g_{j}(\cdot, t) = (g_{j,1}(\cdot, t), \ldots, g_{j,q}(\cdot, t))^\T \in L^2_q(\mathcal{X})$ for all $j = 1, \ldots, p$ and the corresponding bias \remove[ab]{is} $b_t(\cdot) \in L^2_q(\mathcal{X}) $.  {\color{black} Within this context, we define $\theta^\star$ as
$$\theta^\star = \argmin_{t \in \Theta} L\{y_R(\cdot), f(\cdot, t)\} = \argmin_{t \in \Theta}\sum_{k=1}^q \int_\mathcal{X} \left[ y_{R, k}(x) - f_k(x, t)\right]^2  dx,$$
that is, the loss function under consideration is the $L^2_q$-loss and the target parameter is the one which minimizes this loss. By definition, $\theta^\star$ is a stationary point of $L\{y_R(\cdot), f(\cdot, t)\}$. Suppose $g_j(\cdot, t)$ is bounded for all $t \in \Theta$. This allows interchanging a differentiation and an integral.
In the next Proposition, we provide an alternative characterization of this stationarity condition in terms of the bias function extending the result of \cite{plumlee2017bayesian} to multivariate outcomes.
\begin{proposition}\label{prop:multivariate_conditions}
    Consider the loss $L\{y_R(\cdot), f(\cdot, t)\}$ defined above. Recall $g_{j,k}(x, t) = \frac{\partial}{\partial t_j} f_k(x, t)$. An equivalent condition for $\frac{\partial}{\partial t} L\{y_R(\cdot), f(\cdot, t)\}\rvert_{t =\theta^*} = 0$ is     \begin{equation}\label{eq:orthogonal_bias}
    \int_{\mathcal{X}} \sum_{k=1}^ q g_{j,k}(x, \theta^*) b_{\theta^*, k}(x)dx = 0, \quad \text{for all } j = 1, \ldots, p.
    \end{equation}
\end{proposition}
}
{\color{black} The condition reduces to \eqref{eq:plumee_orthogonality} when $q = 1$. Given this condition, the space of possible bias functions reduces to ones that satisfy \eqref{eq:orthogonal_bias}. Naturally, Bayesian inference can proceed by eliciting a prior distribution on bias function that satisfy this stationarity condition. The key issue is that $\theta^\star$ is unknown. Suppose $\tilde{\theta}$ is our guess about the parameter combination at which the computer model best approximates the system given the observed data and the computer model. A data-dependent choice of $\tilde{\theta}$ involves minimizing the corresponding empirical risk.
}
\begin{equation}
    \tilde{\theta}  = \argmin_{t \in \Theta} \frac{1}{nq}\sum_{k=1}^ q\sum_{i=1}^{n} \{y_{F,k}(x_i) - f_k(x_i, t)\}^2.
\end{equation}
Validity of this definition is provided in the following result where we establish $\theta^*$ as the minimizer of population risk and then leverage on empirical risk minimization results to show that $\tilde{\theta}$ converges to $\theta^*$ for large $n$. 
\begin{proposition}\label{prop:data_est}
Let $P = P_{y\mid x} P_x$ denote the data generating distribution where $P_x$ is the uniform measure on $\mathcal{X}$. Suppose that $\Theta$ is compact and that the $k$-th computer model is Lipshitz in $t \in \Theta$ uniformly over $x \in \mathcal{X}$, that is, there exists $L_k>0$ such that $|f_k(x, t_1) - f_k(x, t_2)|\leq L_k \norm{t_1 - t_2}_2^2$ for $k = 1,\ldots, q$. We also assume that the population risk $L\{y_R(\cdot), f(\cdot, t)\}$ has a unique minimizer that is well-separated, namely $\argmin_{t \in \Theta}\\ L\{y_R(\cdot), f(\cdot, t)\}$ is uniquely minimized at $\theta^\star$
and for any $\epsilon > 0, \inf_{t: ||t -\theta^\star||\geq \epsilon} L\{y_R(\cdot), f(\cdot, t)\} >  L\{y_R(\cdot), f(\cdot, \theta^\star)\} $.
Then $\tilde{\theta} \overset{P}{\to} \theta^*$. \note[ab]{\cred ** Again, should this be $t^*$? Also, better to call it $\theta^*$ throughout. The current notation can be rather confusing. Update: should be $\theta^*$. -- AB (1/30)}
\end{proposition}
\noindent {\color{black} When $P_x$ is not uniform over $\mathcal{X}$, the definition of the loss should be changed accordingly, e.g. $L\{y_R(\cdot), f(\cdot, t)\} = \int \{y_R(x) - f(x, t)\}^2 dP_x$. In that case, the empirical risk minimization should be carried over the empirical distribution $P_{n,x}$ instead. This essentially leads to a weighted least-squares problem without changing the conclusion of the above result.} The compactness of $\Theta$ \change[ab]{although appears to be a technical condition, it is often practical in many calibration problems}{is a technical assumption which nevertheless is often practical in many calibration problems}. {\color{black} Since we fix $\tilde{\theta}$ before eliciting prior distributions on $(\theta, b)$, our approach is not entirely Bayesian in the classical sense, but is of a modular Bayes nature \citep{bayarri2007framework}.}


Having defined the point at which we want the bias function $b(x)$ to satisfy the orthogonality condition \eqref{eq:orthogonal_bias}, we now proceed to define the projection posterior. {\color{black} Define $\mathcal{F}_{\tilde{\theta}} = \{b_{\tilde{\theta}}: \int \sum_{k=1}^q g_{j,k}(x,\tilde{\theta})b_{\tilde{\theta}, k}(x) dx = 0, j = 1, \ldots, p\}$ as the set of bias functions that satisfy the orthogonality condition at $\tilde{\theta}$.} Instead of defining a prior distribution on $b_{\tilde{\theta}} \in  \mathcal{F}_{\tilde{\theta}}$, our strategy is to start with a generic prior and then project the posterior distribution to $\mathcal{F}_{\tilde{\theta}}$. 
Standard Hilbert space theory implies that $\mathcal{F}_{\tilde{\theta}}$ is a non-empty, closed and convex subset of $L^2_q(\mathcal{X})$. Then by the Hilbert projection theorem \citep{tsiatis2006semiparametric}, there exists a unique projection $b_{\tilde{\theta}}^*$ of any $b_{\tilde{\theta}} \in L^2_q(\mathcal{X})$. Define the projection as 
 \begin{equation}\label{eq:proj_operator}
    T_{\mathcal{F}_{\tilde{\theta}}}(b_{\tilde{\theta}}) = \{b_{\tilde{\theta}}^* \in \mathcal{F}_{\tilde{\theta}} : \norm{b_{\tilde{\theta}}^* - b_{\tilde{\theta}}}_{L^2_q} = \inf_{b \in \mathcal{F}_{\tilde{\theta}}} \norm{b - b_{\tilde{\theta}}}_{L^2_q}\}.
\end{equation}


We now describe how $T_{{\mathcal{F}}_{\tilde{\theta}}}$ is used to define the projection posterior. Suppose a prior distribution $\Pi(\theta, b) = \Pi(\theta)\Pi(b)$ is elicited on $\Theta \times L^2_q(\mathcal{X})$. 
Let $\widetilde{B} = \widetilde{B}_1 \times \widetilde{B}_2$ be a measurable subset of the Borel $\sigma$-algebra of $\Theta \times {\mathcal{F}}_{\tilde{\theta}}$ . Then given $y^{(n)}$, we define the posterior probability $\Pi_{\text{proj}}(\widetilde{B}\mid y^{(n)})$ of $\widetilde{B}$ under the prior $\Pi(\theta, b)$ as $\Pi(B \mid y^{(n)})$ where $B = \widetilde{B}_1 \times  T_{\mathcal{F}_{\tilde{\theta}}}^{-1}(\widetilde{B}_2)$ where $T_{\widetilde{\mathcal{F}}_{\theta}}^{-1}(\widetilde{B}_2) = \{b_{\tilde{\theta}}: T_{\mathcal{F}_{\tilde{\theta}}}(b_{\theta}) \in \widetilde{B}_2 \}$, that is,
\begin{equation}\label{eq:posterior_projection}
    \Pi_{\text{proj}}(\widetilde{B}\mid y^{(n)}) = \Pi(B \mid y^{(n)}) = \dfrac{\int_{B} \ell(y^{(n)}\mid \theta, b) d\Pi(\theta, b)}{\int \ell(y^{(n)}\mid \theta, b) d\Pi(\theta, b)}.
\end{equation}
Two remarks are in order here.
{ \color{black} 
\begin{rem}
    We note here that our approach is significantly different from a related projection based method by \cite{xie2021bayesian} where the authors consider a Gaussian process prior on the bias function. The calibration parameter, is then treated as a functional of the bias function. Hence, a prior is induced on the calibration parameter through the Gaussian process prior on the bias function. Although this serves the statistical issue of identifiability, the user loses control over the prior specification for the calibration parameters.
\end{rem}
\begin{rem}
   The key benefit of the proposed approach is that it enables decoupling the inference on the bias function and the calibration parameter even under the orthogonality constraint. This in turn allows flexible prior specification on both parameters. In particular, the user is free to use a broader class of non-stationary priors for the bias function, such as BMARS, BART, etc. In contrast, the approach of \cite{plumlee2017bayesian} allows non-stationarity in a very specific manner.
\end{rem}
}
Measurability of $T$ is guaranteed since $\mathcal{F}_{\tilde{\theta}}$ is non-empty, closed and convex; see also \cite{sen2018constrained}. 
Although \eqref{eq:proj_operator} is defined as a solution to an optimization problem, in this particular case, it has an explicit form which we call functional projection.
\begin{lemma}\label{lm:projection_formula}
Fix $b \in L^2_q(\mathcal{X})$ and let $\tilde{\theta}$ be defined as above. Let the functions $g_{j,k}(x),\, j = 1, \ldots, p$ satisfy the following: $\sum_{j=1}^p \alpha_j g_{j,k}(x) = 0$ for all $x \in \mathcal{X}$ and for all $k = 1, \ldots, q$ iff $\alpha_j =0, \, j = 1, \ldots, p$. Then $T_{\mathcal{F}_{\tilde{\theta}}}(b) = b^*(x) = b(x) - \sum_{j=1}^p \lambda_j^* g_j(x)$ where the vector $\lambda = (\lambda_1^*, \ldots,\lambda_p^*)^\T$ satisfies $Q\lambda = \eta$ with $\eta = (\sum_{k=1}^q\langle  b_k, g_{1,k}\rangle, \ldots, \sum_{k=1}^q \langle  b_k ,g_{p,k} \rangle)^\T$ and $Q$ is a $p \times p$ matrix with elements $Q_{jj'} = \sum_{k=1}^q \langle g_{j,k}, g_{j',k} \rangle$.
\end{lemma}

{\color{black} In practice, implementing the functional projection requires evaluating the gram matrix $Q$ and $\eta$ which in turn involve evaluating integrals of the form $\langle g_{j,k}, g_{j', k}\rangle$ and $\langle b_k, g_{j,k}\rangle$. When the dimension of $\mathcal{X}$ is less than or equal to two, we use Gauss-Legendre quadrature to compute these quantities. Specifically, let $(x_t, w_t)_{t=1}^m$ denote the quadrature points and their associated weights, respectively. Then $\langle g_{j,k} , g_{j', k}\rangle \approx \sum_{t=1}^m g_{j,k}(x_t) g_{j', k}(x_t) w_t$. To compute $\langle b_k, g_{j,k} \rangle$ we use a similar strategy. The only difference in this case is that we use the predicted value of $b_{\tilde{\theta}}$ on $\{x_t\}_{t=1}^m$ to approximate $\langle b_k, g_{j,k} \rangle$ by $\sum_{t=1}^m b_k(x_t) g_{j', k}(x_t) w_t$. These predicted values are automatically obtained from the posterior predictive of $b_{\tilde{\theta}}$. Naturally, the number of quadrature points $m$ is a tuning parameter here. In our experiments, the choice $m = 15$ for $\mathcal{X} \subset \mathbb{R}$ yielded robust results. When the input dimension is higher than 2, quadrature rules require a large number of points. In such cases, we use Monte Carlo to approximate these integrals. For example, let $\mathcal{X} = [0,1]^5$. Then $\langle g_{j,k} , g_{j', k} \rangle= \int_{\mathcal{X}} g_{j,k}(x, \tilde{\theta}) g_{j', k}(x, \tilde{\theta}) dx = \mathbb{E}_{X \sim \pi_U} [g_{j,k}(X, \tilde{\theta})g_{j', k}(X, \tilde{\theta})]$, where $\pi_{U}$ is the product Uniform distribution over $[0,1]^5$. The corresponding estimate is  $M^{-1} \sum_{i=1}^M g_{j,k}(X_i, \tilde{\theta}) g_{j'k}(X_i, \tilde{\theta}) $ where $X_i \overset{iid}{\sim} \pi_U$ for $i = 1, \ldots, M$. The same technique can be used to implement the orthogonal calibration method with a modified covariance function \citep[Equation 6]{plumlee2017bayesian}. However, a key difference is that for this method, an integral over $\mathcal{X}^2$ is required to compute the covariance kernel. Moreover, to ensure positive definiteness of the kernel, an appropriately large number of Monte Carlo samples need to be selected which sharply increases the computational cost. We include a high-dimensional example in the numerical section to investigate the performance of the projection approach in such situations.}


Having defined the projection for our purpose, we can then devise an MCMC algorithm to sample from $\Pi_{\text{proj}}(\cdot \mid y^{(n)})$ defined in \eqref{eq:posterior_projection}.
\begin{algorithm}
\caption{Projection sampler 1}
1. Update $b_{\tilde{\theta}} \sim \Pi(b_{\tilde{\theta}} \mid \theta, y^{(n)})$

2. Project $b_{\tilde{\theta}}$ onto $\mathcal{F}_{\tilde{\theta}}$ following Lemma \ref{lm:projection_formula} to obtain $b^*_{\tilde{\theta}}$.

3. Update $\theta \sim \Pi(\theta \mid b^*_{\tilde{\theta}}, y^{(n)})$.

\label{algo:conditional_sampler}
\end{algorithm}

As an example, consider the case when $b_{\tilde{\theta}}$ is endowed with a zero mean Gaussian process prior with covariance kernel $C(\cdot, \cdot)$. Then following Algorithm \ref{algo:conditional_sampler}, we update $b_{\tilde{\theta}} \sim \Pi(b \mid \tilde{\theta}, y^{(n)})$ which is a multivariate Gaussian distribution. Next, we compute the corresponding projection $b^*_{\tilde{\theta}}$ using Lemma \ref{lm:projection_formula}. Finally, we sample $\theta \sim \Pi(\theta \mid b^*_{\tilde{\theta}}, y^{(n)})$. This strategy also works for other priors such as BART, where in the first step we update the parameters of the BART model using their respective full conditionals. The next steps are exactly the same.  An alternative finite-dimensional projection method is described in the supplementary document.

 For a full Bayesian inference on \eqref{eq:kennedy_ohagan} under the projection posterior framework described above, one should ideally elicit a prior $\Pi(\Sigma_F)$ on the unknown covariance matrix $\Sigma_F$. However, this often adds to computational challenges already present in a calibration problem and may not influence the uncertainty observed in $\theta$ - the central goal of these problems \citep{bayarri2007computer}.  As an alternative, a modular Bayes approach is considered here where we use a plug-in estimate $\hat{\Sigma}_F$. This estimate is constructed by fitting a nonparametric regression model, i.e. $y_{F,k}(x) = y_{R,k}(x) + \epsilon_k(x)$ to each outcome separately with idiosyncratic variance parameters $\sigma_{F,k}^2$ with conjugate Inverse Gamma priors and $y_{R,k} \sim \Pi_k$ for some nonparametric prior on $y_{R,k}$. We then set $\hat{\Sigma}_F = \text{Cov}(E)$ where $E$ is the error matrix with columns $E_{k} = y_{F,k} - \bar{y}_R$ where $\bar{y}_R$ is the posterior mean of predictions on the training set. Given $\hat{\Sigma}_F$, it can be assumed without loss of generality that $\epsilon(x) \sim \Gauss(0, \mathrm{I}_q)$ in \eqref{eq:kennedy_ohagan}. This also makes prior elicitation on the multivariate bias $b$ simpler; a natural choice is $\Pi(b) = \prod_{k=1}^q \Pi(b_k)$.

\section{Explicitly unavailable computer model}\label{sec:unknown_computer_model}
Until now we have assumed that the computer model $f(x,t)$ is either explicitly specified or can be evaluated cheaply using a code for every $(x,t) \in \mathcal{X} \times \Theta$. However, in \change[ab]{most}{many} calibration problems, including ours, this \add[ab]{is} not the case; the computer code \change[ab]{is}{can be} computationally very expensive; see also Section \ref{sec:schwarzchild_model}. \note[ab]{This could be a good place to remind/draw attention to Section 2, where you mention this.} The standard approach in the literature has been to approximate the computer model using a nonparametric method, typically the Gaussian process \citep{santner2003design}. While the original calibration framework developed by \cite{kennedy2001bayesian} modeled the observed data and the computer model simultaneously, a modular approach to inference has been advocated by many authors, e.g. \cite{bayarri2007framework, plumlee2017bayesian} wherein the modeling of the explicitly defined computer model is done separately from the modeling of the observed data. Here, we adopt this modular approach with one crucial difference with \cite{plumlee2017bayesian} where the author uses a probabilistic definition of the computer model using a Gaussian process. We, on the other hand, use a deterministic definition $\hat{f}$ of the computer model obtained using any nonparametric approximator such as the posterior mean of a BART fit, a neural net obtained by minimizing the least square error between code evaluations and the neural net output, or a random forest fit; such an approach was also adopted by \cite{xie2021bayesian}. Our motivation for doing so is to allow for more computational and modeling flexibility than a Gaussian process framework. Specifically, we work with the posterior mean of a BART fit which proved to be superior among all other choices of models for $f(x,t)$ in our numerical experiments as well as in our motivating example. We compared the squared error loss on a held-out test data of size $n_t = 50$ for these methods: $(n_t)^{-1} \sum_{i = 1}^{n_t} \{f(x, t_i) - \hat{f}(x, t_i)\}^2$ for each location $x$ in the observed data $x \in \mathcal{X}$ and each output of the computer model. In the best case scenario, the ratio of the squared error loss for the posterior predictive mean improved over the second best method (Random Forest) by a factor of almost 5. {\color{black} We also use this posterior mean to approximate the gradient of the computer model using the two point estimate $\frac{\partial}{\partial \theta} \hat{f}(x, \theta) \approx \{\hat{f}(x, \theta + h) - \hat{f}(x, \theta - h )\}/2h$ for small positive $h$.}

\section{Simulations}\label{sec:simulations}
{\color{black} In this section, we compare our method to the GP orthogonal calibration by \cite{plumlee2017bayesian} (OGP), the projected calibration (PCAL) method by \cite{xie2021bayesian}, and the scaled Gaussian process (S-GaSP) method by \cite{gu2018scaled} in several simulation experiments.} Since these methods were originally developed for one outcome, we consider the case $q =1$. 
We focus on each of these methods' ability to estimate $\theta^*$, uncertainty quantification, and their associated computing time. {\color{black} All reported computing times correspond to a 3.8 GHz 8-Core Intel Core i7 machine. } 
\subsection{Explicitly available computer model}
For generating the data, we consider {\color{black} three } situations when the definition of the computer model is available explicitly:
\begin{enumerate}
    \item {\bf Model 1:} $y_R(x) = 4x + x\sin 5x$, $f(x, t) = tx$ where $x \in [0,1]$ where $\theta^* = 3.56$ \cite[Example 5.1]{plumlee2017bayesian},
    \item {\bf Model 2:} $f(x, t) = 7\{\sin (2\pi t_1 - \pi)\}^2 + 2\{(2\pi t_2 - \pi)^2 \sin (2\pi x - \pi)\}$, $y_R(x) = f(x, \theta^*)$ where $\theta^* = (0.2, 0.3)$, $x \in [0,1]$ \citep[Configuration 1]{xie2021bayesian}.
     \item {\bf Model 3:} $f(x, t) = t_1\sin x_1 x_2 + t_2(x_3 - 0.5)^2 + 10 x_4 + 5x_5$, $y_R(x) = f(x, \theta^*)$ where $\theta^* = (10, 20)$ (Friedman function). Here $x_j \in [0,1]$ for all $j = 1, \ldots, 5$.
\end{enumerate}

For all cases we simulate $n = 100$ field observations from the model $y^F(x_i) = y^R(x_i) + \epsilon_i, \, \epsilon_i \sim \Gauss(0, 0.2^2)$ where we sample the $x_i$'s uniformly over $[0,1]$. For the proposed method, we consider two priors on the bias function - GP and BART. {\color{black} For the GP prior, we use the Matern covariance kernel defined as $C(x, x') = \phi^2 (2^{1-\nu}/\Gamma(\nu))(\sqrt{2\nu} |x-x'|/\psi)^\nu K_\nu(\sqrt{2\nu} |x-x'|/\psi)$, where $K_\nu(\cdot)$ is the modified Bessel function of the second kind. We set $\phi =1$, $\psi = 1/2$ and consider three different choices of $\nu = 1/2, 3/2, 5/2$. We choose three different values of $\nu$ to understand the effect of the prior smoothness assumption of the bias on the posterior distribution of the calibration parameter. These three different covariance functions are also used for OGP. For the S-GaSP and PCAL method we choose the Matern covariance $\phi = 1$, $\psi = 1/2$ and $\nu = 3/2$. For the calibration parameter $\theta$, we choose a Gaussian prior centered at 0 with standard deviation 10 - this choice of prior is implemented for all versions of the PGP, PBART and OGP. We could not implement the S-GaSP method for {\bf Model 2} in our repeated attempts.}

{\color{black} To sample $\theta$, we use a random-walk Metropolis sampler with target acceptance rate 0.3 for all versions of PGP and PBART. For example, consider the model $y = f(\theta, x) + b_{\tilde{\theta}}(x) + \epsilon$. Given a $\theta$, we propose a candidate $\theta_c \sim \Gauss(\theta, \sigma_{MH}^2)$. Let $f_\theta = (f(\theta, x_1), f(\theta, x_2, ), \ldots, f(\theta, x_n))^\T$ and $f_{\theta_c} = (f(\theta_c, x_1), f(\theta_c, x_2, ), \ldots, f(\theta_c, x_n))^\T$. Then conditional on $b_{\tilde{\theta}}(x)$, $y^{(n)} - b_{\tilde{\theta}} \sim \Gauss(f_\theta, \sigma^2 \mathrm{I}_n)$ at the current value $\theta$. Similarly, $y^{(n)} - b_{\tilde{\theta}} \sim \Gauss(f_{\tilde{\theta^\star}}, \sigma^2 \mathrm{I}_n)$. This is used to compute the Metropolis--Hastings ratio to update $\theta$.
In low-dimensions, the sampler mixes quite well as is evidenced in Figure \ref{fig:theta_posterior}, where we plot the posterior samples of $\theta$ for {\bf Model 1} and of $(\theta_1, \theta_2)$ for {\bf Model 2}, under a PBART posterior. A red line is added in each figure to indicate the true values of the parameters. We also computed the effective sample size (ESS) $L_{eff} =  L/(1 + 2\sum_{k=1}^\infty r_k)$ with $L$ the length of the chain, and $r_k$ the autocorrelation of the chain with lag $k$. For {\bf Model 1} and PBART, the effective sample size with 5000 MCMC iterations is on average 850. For {\bf Model 2} these numbers were approximately 750 and 940 for $(\theta_1, \theta_2)$.  } 

{\color{black} In Tables \ref{tab:case1} and \ref{tab:case2}, we summarize the results for {\bf Model 1} and {\bf Model 2} over 100 replications. Specifically, we report the range of posterior mean, range of standard deviation, and coverage of $95\%$ credible intervals along with the average runtime of the methods over these 100 replications. Although all methods generally agree in terms of the posterior mean, a clear distinction emerges in the spread of the posterior distribution. In particular, our posterior projection based approaches have lower posterior standard deviation compared to OGP and S-GaSP (for {\bf Model 1}). For example, posterior standard deviations of PGP, PBART, PCAL are almost an order of magnitude smaller compared to OGP and S-GaSP for {\bf Model 1}, and about 20-30\% smaller for {\bf Model 2}. Within the projection based approaches,  PCAL and PBART achieve the lowest standard deviation both for {\bf Model 1} and {\bf Model 2}. Moreover, the reduced posterior spread does not come at the cost of mis-calibration of coverage -- PGP, PBART, PCAL provide close to the nominal coverage in all settings. Among the two posterior projection approaches considered here, PGP and PBART, PGP provides slightly better uncertainty quantification. This might be due to the non-smooth prior specification resulting from BART. A smooth version of BART such as \cite{linero2018bayesian} might potentially address this issue.
In terms of computing time, we see that PBART achieves fastest computing time.
}
\begin{figure}
     \centering
     \begin{subfigure}[b]{0.3\textwidth}
         \centering
         \includegraphics[width=\textwidth, height = 4cm]{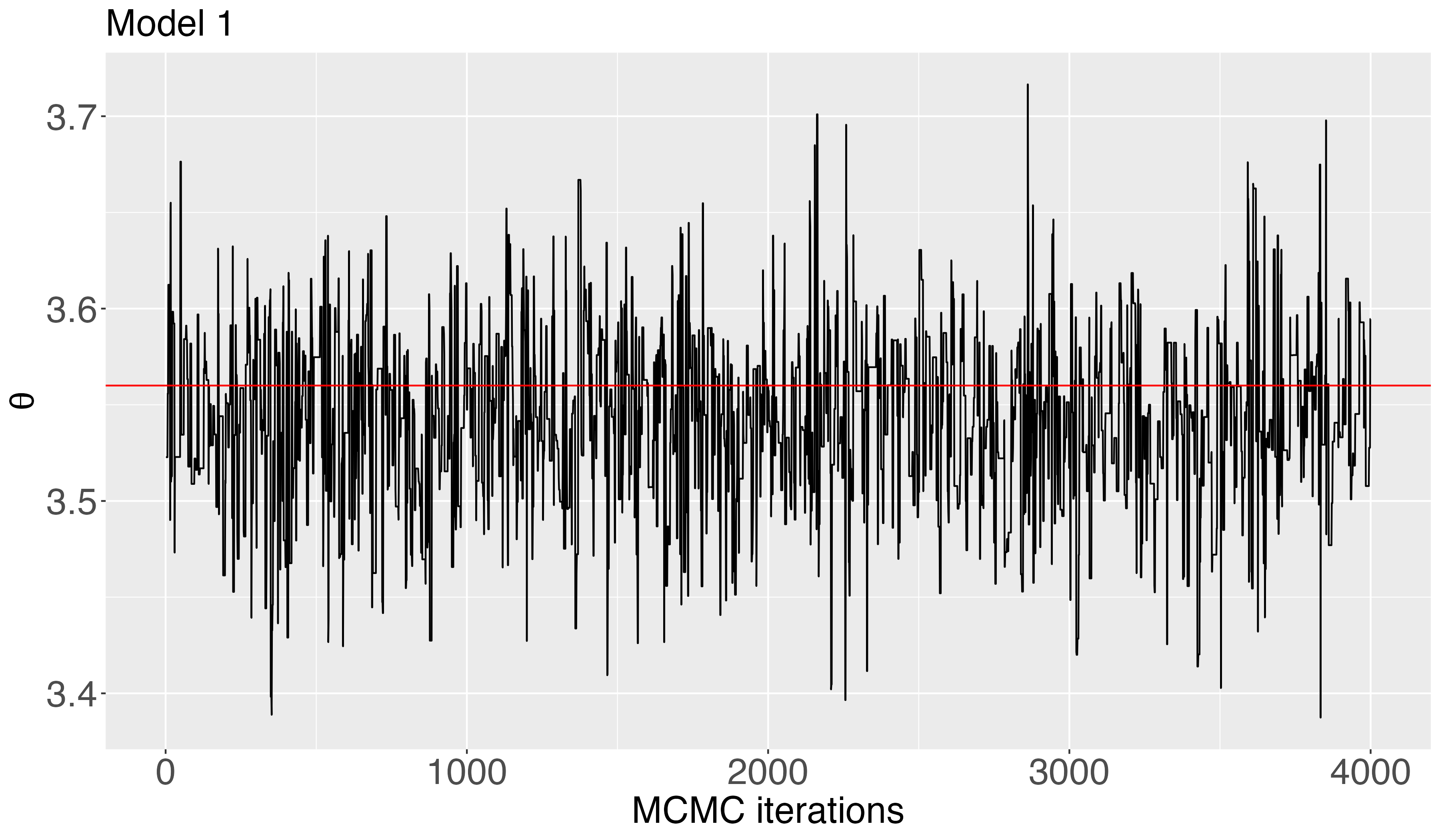}
         \caption{Traceplot of posterior samples of $\theta$ for {\bf Model 1}. }
     \end{subfigure}
     \hfill
     \begin{subfigure}[b]{0.3\textwidth}
         \centering
         \includegraphics[width=\textwidth, height = 4cm]{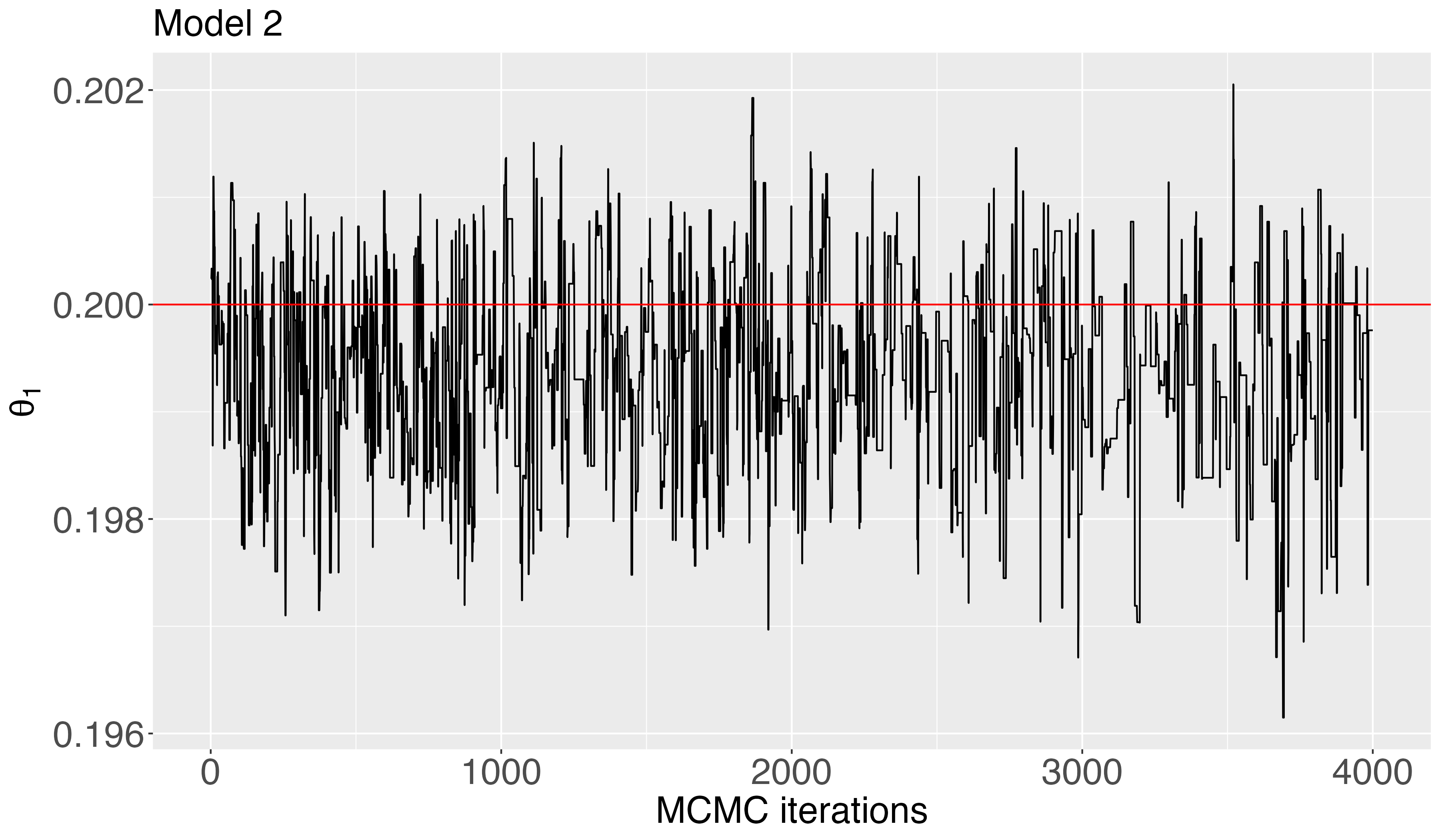}
         \caption{Traceplot of posterior samples of $\theta_1$ for {\bf Model 2}.}
     \end{subfigure}
     \hfill
     \begin{subfigure}[b]{0.3\textwidth}
         \centering
         \includegraphics[width=\textwidth, height = 4cm]{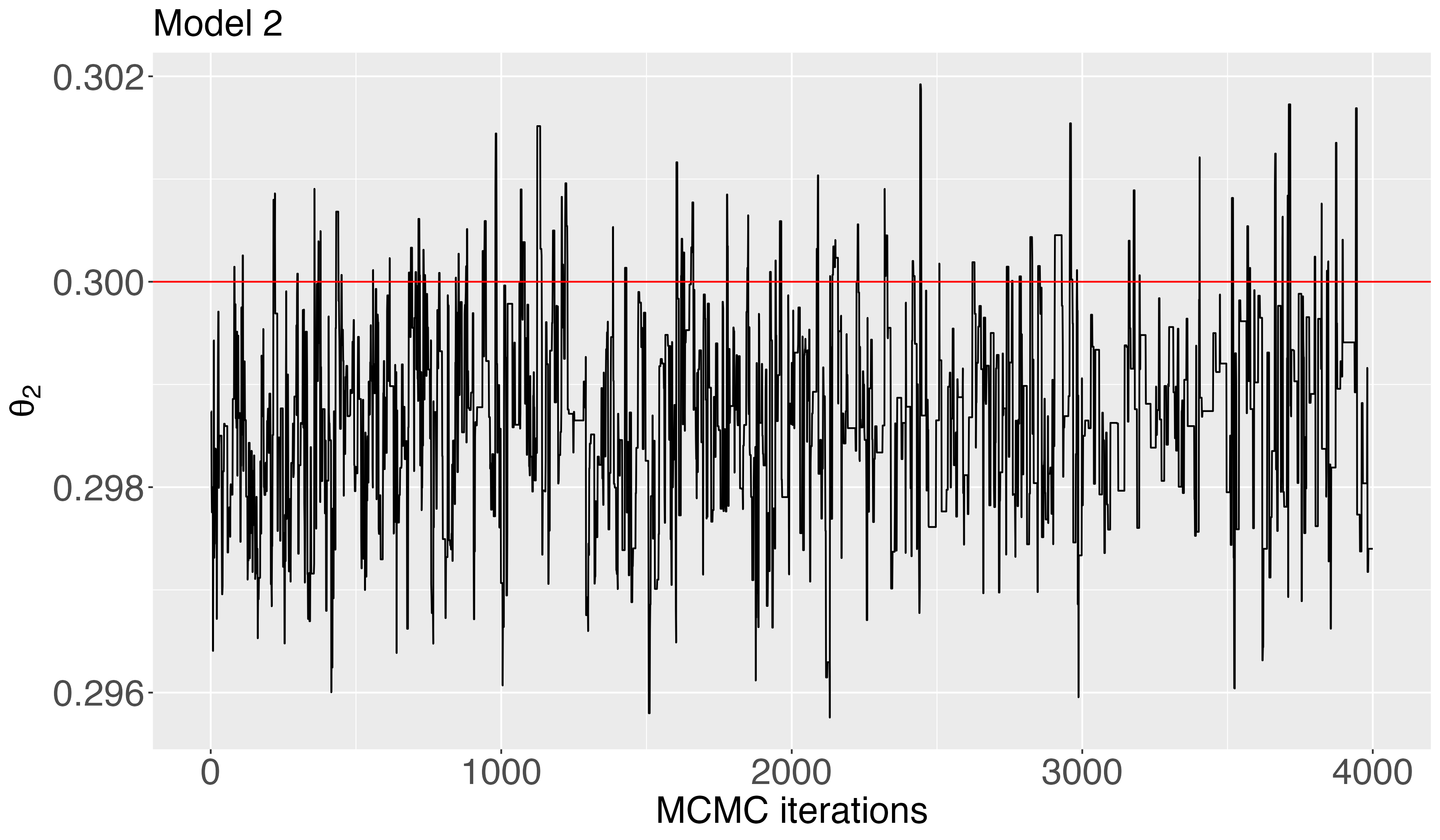}
         \caption{Traceplot of posterior samples of $\theta_2$ for {\bf Model 2}.}
     \end{subfigure}
        \caption{PBART posterior traceplots for {\bf Model 1} and {\bf Model 2}.}
        \label{fig:theta_posterior}
\end{figure}
\begin{table}
    \centering
    \scalebox{0.6}{
    \begin{tabular}{cccccccccc}
    \hline
         & \multicolumn{9}{c}{\textbf{Model 1 ($\theta^* = 3.56$)}} \\
         \hline
         & PGP(1/2) & PGP(3/2) & PGP(5/2) & PBART & OGP(1/2) &  OGP(3/2) & OGP(5/2) & S-GaSP  & PCAL\\
         \hline
         Mean & (3.47, 3.66) & (3.48, 3.65) & (3.49, 3.65) & (3.48, 3.62) & (3.49, 3.62) & (3.49, 3.62) & (3.41, 3.68) & (3.46, 3.68) & (3.48, 3.67) \\
        Std. Dev. & (0.09, 0.11) & (0.07, 0.09) & (0.06, 0.08) & (0.04, 0.05) & (0.14, 0.18) & (0.14, 0.18) & (0.14, 0.18) & (0.30, 0.45) & (0.03, 0.04)\\
         Coverage & 0.97 & 0.97 & 0.97 & 0.97 & 0.99 & 0.99 & 0.99 & 1 & 0.95\\
         \midrule
         Runtime &3.78 sec & - & - & 1.88 sec & 9.82 sec & - & - & 3.23 sec & 16.40 min\\
         \hline
    \end{tabular}
    }
    \caption{Simulation results for \textbf{Model 1}. We report the range of posterior mean, range of posterior standard deviation, coverage of the 95\% credible intervals and the average runtime over 100 replications.}
    \label{tab:case1}
\end{table}

\begin{table}
    \centering
    \scalebox{0.55}{
    \begin{tabular}{cccccccccc}
    \hline
         & \multicolumn{8}{c}{\textbf{Model 2 ($\theta^* = (0.2, 0.3)$)}} \\
         \hline
         & & PGP(1/2) & PGP(3/2) & PGP(5/2) & PBART & OGP(1/2) & OGP(3/2) & OGP(5/2) & PCAL \\
         \hline
        \multirow{3}{*}{$\theta_1^*$} &Mean & (0.19, 0.20) & (0.19, 0.20) & (0.19, 0.20) & (0.19, 0.20) & (0.19, 0.20)  & (0.19, 0.20) & (0.19, 0.20) & (0.19, 0.20)\\
        &Std. Dev. & (0.001, 0.001) & (0.001, 0.001) & (0.001, 0.001) & (0.0008, 0.0009) & (0.003, 0.004) & (0.003, 0.004) & (0.003, 0.004)& (0.0007, 0.0007) \\
         & Coverage & 0.93 & 0.93 & 0.93 & 0.93 & 0.97 & 0.97 & 0.97 & 0.96\\
         \midrule
         \multirow{3}{*}{$\theta_2^*$} &Mean & (0.29, 0.30) & (0.29, 0.30) & (0.29, 0.30) & (0.29, 0.30) & (0.29, 0.30)  & (0.29, 0.30) & (0.29, 0.30) & (0.29, 0.30)\\
        &Std. Dev. & (0.001, 0.002) & (0.001, 0.002) & (0.001, 0.002) & (0.0009, 0.001) & (0.004, 0.004) & (0.004, 0.005) & (0.004, 0.005)& (0.0007, 0.0009) \\
         & Coverage & 0.93 & 0.93 & 0.93 & 0.91 & 0.98 & 0.98 & 0.98 & 0.92\\
         \midrule
         & Runtime & 18.18 sec & - & - & 9.04 sec & 15.87 sec & - & -  & 16.58 min  \\
         \hline
    \end{tabular}
    }
    \caption{Simulation results for \textbf{Model 2}. Reported quantities are the same as in Table \ref{tab:case1}.}
    \label{tab:case2}
\end{table}

{\color{black} We next consider the high-dimensional example in {\bf Model 3}. The key challenge in this case is the 5-dimensional input space $\mathcal{X} = [0,1]^5$. To demonstrate the flexibility of the projection based approach, we also consider the BMARS \citep{denison1998bayesian} prior for the bias function, which is also a recursive partitioning based nonparametric prior over functions. We call the resulting projection based approach PBMARS and compare the results with PGP, PBART and OGP. Both the proposed projection-based method and orthogonal calibration method require evaluating integrals over the five-dimensional input space $\mathcal{X}$. To mitigate the curse of dimensionality associated with quadrature-based methods, we use a Monte Carlo estimate of the integrals with Monte Carlo sample size 200 for PGP, PBART, PBMARS and 500 for OGP. The results of 100 replications are summarized in Table \ref{tab:case3}. The Gaussian process based methods (both the projection approach and OGP) have much higher computation times than PBART and PBMARS. Additionally, in accord with our results for {\bf Model 1} and {\bf Model 2}, we see that both PBART and PBMARS have the smallest posterior standard deviation in this case as well. All methods show significantly large uncertainty in $\theta_2$ compared to $\theta_1$.
The mixing of the MCMC samples for $(\theta_1, \theta_2)$ are shown in Figure \ref{fig:theta_posterior2} for PBART. The chain for $\theta_1$ mixes quite well but the chain for $\theta_2$ shows high correlation. We also observed this same effect for the other methods. A potential improvement can be achieved by gradient-based MCMC methods such as the HMC \cite{neal2011mcmc}.

}

\begin{table}
    \centering
    \scalebox{0.7}{
    \begin{tabular}{cccccccc}
    \hline
         & \multicolumn{6}{c}{\textbf{Model 3 ($\theta^* = (10, 20)$)}} \\
         \hline
         & & PGP(1/2) & PGP(3/2) & PGP(5/2) & PBART &  PBMARS  & OGP(5/2) \\
         \hline
        \multirow{3}{*}{$\theta_1^*$} &Mean & (9.96, 10.03) & (9.95, 10.05) & (9.95, 10.05) & (9.95, 10.05) & (9.95, 10.04) & (9.95, 10.04)\\
       & Std. Dev. & (0.04, 0.05) & (0.04, 0.05) & (0.04, 0.05) & (0.02, 0.03) & (0.02, 0.03) & (0.05, 0.13) \\
       & Coverage & 0.94 & 0.94 & 0.94 & 0.91 & 0.91 & 0.99 \\
       \hline
       \multirow{3}{*}{$\theta_1^*$} &Mean & (19.77, 20.51) & (19.43, 20.42) & (19.57, 20.28) & (19.48, 20.38) & (19.50, 20.38) & (19.31, 20.84)\\
       & Std. Dev. & (0.29, 0.37) & (0.26, 0.34) & (0.24, 0.35) & (0.20, 0.26) & (0.18, 0.25) & (0.51, 0.88) \\
       & Coverage & 0.94 & 0.93 & 0.93 & 0.91 & 0.91 & 1 \\
         & Runtime &1.72 min & - & - & 13.16 sec & 10.01 sec & 59.13 sec\\
         \hline
    \end{tabular}
    }
    \caption{Simulation results for \textbf{Model 3}. We report the posterior mean, posterior standard deviation, intervals and the average runtime.}
    \label{tab:case3}
\end{table}
\begin{figure}
     \centering
     \begin{subfigure}[b]{0.4\textwidth}
         \centering
         \includegraphics[width=\textwidth, height = 4cm]{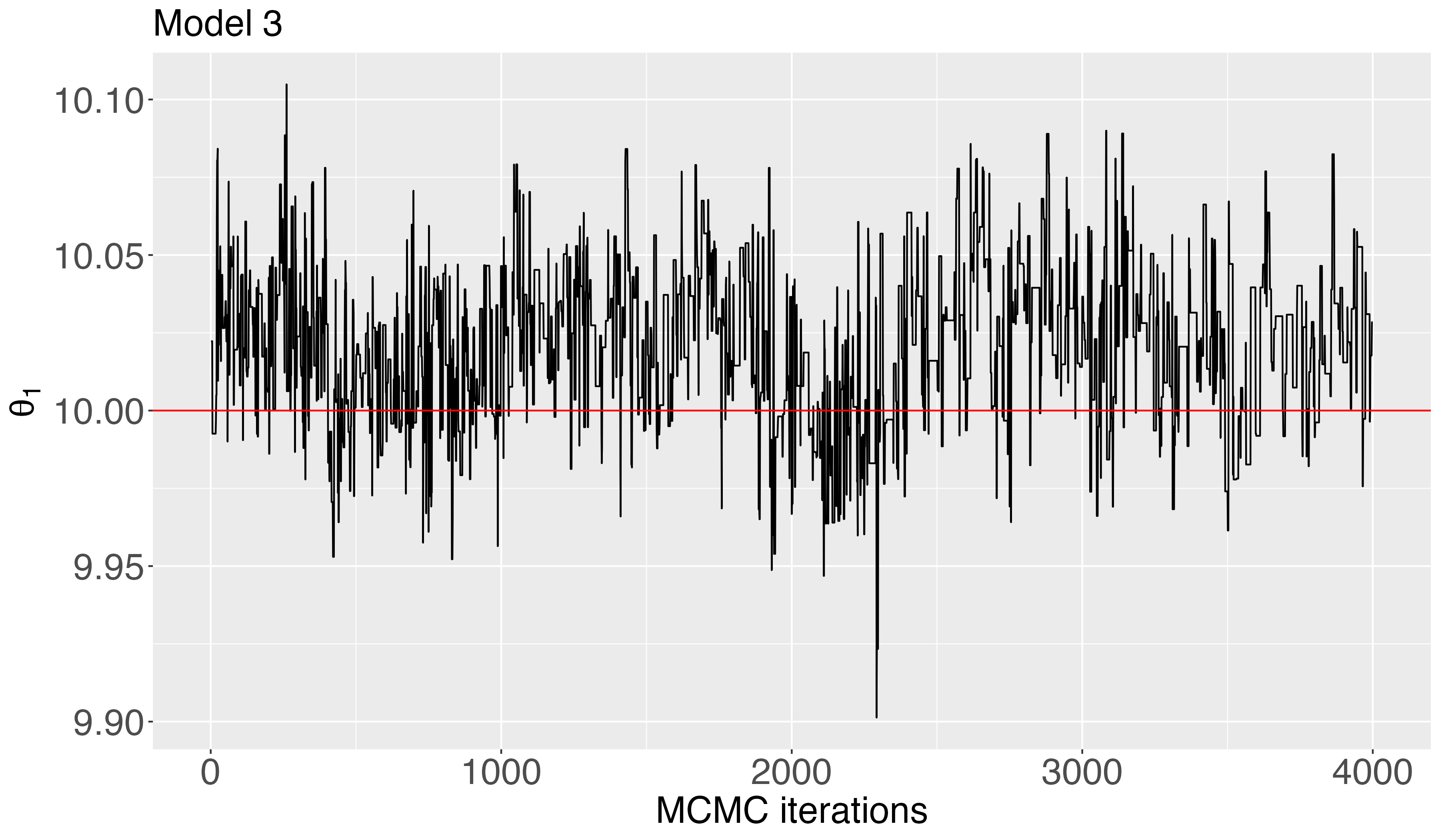}
         \caption{Traceplot of posterior samples of $\theta_1$ for {\bf Model 3}. }
     \end{subfigure}
     \begin{subfigure}[b]{0.4\textwidth}
         \centering
         \includegraphics[width=\textwidth, height = 4cm]{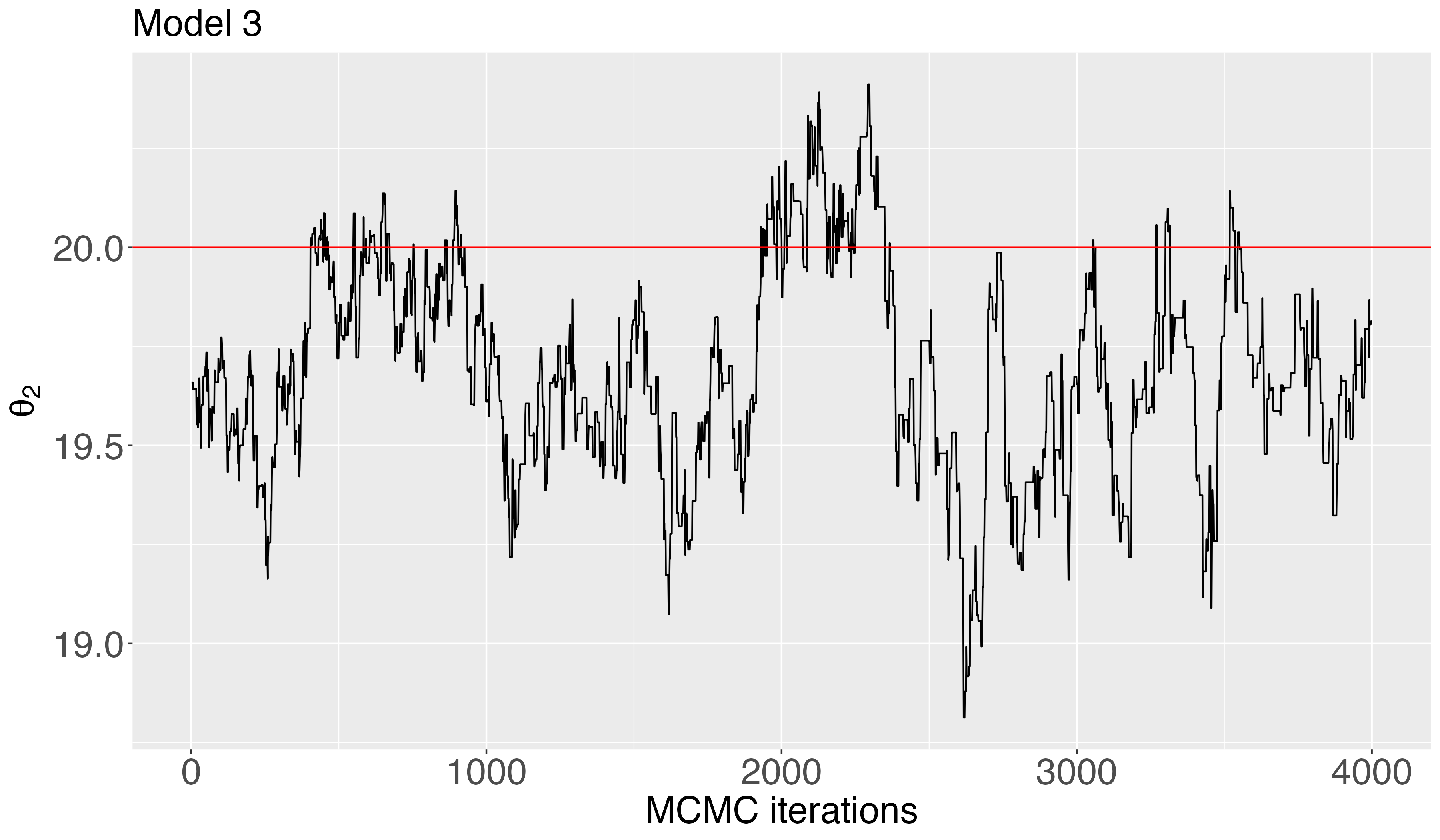}
         \caption{Traceplot of posterior samples of $\theta_2$ for {\bf Model 3}.}
     \end{subfigure}
        \caption{PBART posterior traceplots for {\bf Model 3}.}
        \label{fig:theta_posterior2}
\end{figure}

\subsection{Explicitly unavailable computer model}
We additionally consider one more situation where we do not assume the knowledge of $f(x,t)$ in \textbf{Model 2} and approximate it by the posterior predictive mean of a BART when $N = 7^2$ code outputs are available on the observed values of $x$, meaning we have for each value $\{\theta_1, \ldots, \theta_N\}$ the computer code output is available for $f(\theta_k, x_i), \, k = 1, \ldots, N, i = 1, \ldots, n$. Given the estimate $\hat{f}$, we approximate $\frac{\partial}{\partial \theta_1}f(\theta, x)$ and $\frac{\partial}{\partial \theta_2}f(\theta, x)$ by the two-point estimate described in Section \ref{sec:unknown_computer_model}. The results for PGP and PBART are summarized in Table \ref{tab:case4}. The key difference from the results in Table \ref{tab:case2} is the increased uncertainty in the estimates of $(\theta_1, \theta_2)$.
\begin{table}
    \centering
    \scalebox{0.8}{
    \begin{tabular}{ccc}
    \hline
         & \multicolumn{2}{c}{\textbf{Model 2 ($\theta^* = (0.2, 0.3)$)}} \\
         \hline
         & PGP & PBART  \\
         \hline
         Mean &  (0.21, 0.33) & (0.22, 0.34)  \\
        Std. Dev. & (0.03, 0.03) & (0.04, 0.03)\\
         Runtime & 24 min & 20 min\\
         \hline
    \end{tabular}
    }
    \caption{Simulation results for \textbf{Model 3}. We report the posterior mean, posterior variance and the average runtime.}
    \label{tab:case4}
\end{table}


\section{Application to Schwarzschild's model}\label{sec:schwarzchild_model}
The study of the mass distribution within galaxies is central in the quest of understanding of black holes, stellar components, dark matter, and the growth of galaxies. \cite{schwarzschild1979numerical} presented an orbital superposition method for constructing a self-consistent mass model of galaxies. It consists of ``integrating a representative set of orbits in a static triaxial gravitational potential, and finding weights for these orbits such that their superposition reproduces the assumed mass distribution" \citep{quenneville2021dynamical}. Besides determining black hole masses, Schwarzschild modeling is a powerful tool for measuring a galaxy's mass-to-light ratio, dark matter halo properties, stellar orbital distribution, viewing orientation, and intrinsic three-dimensional shape, allowing for the further study of galaxy assembly histories \citep{mehrgan2019}. 
Several modifications to the initial method have since been proposed among which \cite{van2008triaxial}'s triaxial orbit superposition has become very popular. In standard applications of the method, the model is compared with kinematic and photometric data to determine best-fit parameters such as black hole mass $(\theta_1)$, stellar mass-to-light ratio $(\theta_2)$, and fraction of dark matter halo $(\theta_3)$.

\begin{figure}
     \centering
     \begin{subfigure}{0.48\textwidth}
         \centering
         \includegraphics[height = 3cm, width = 8cm]{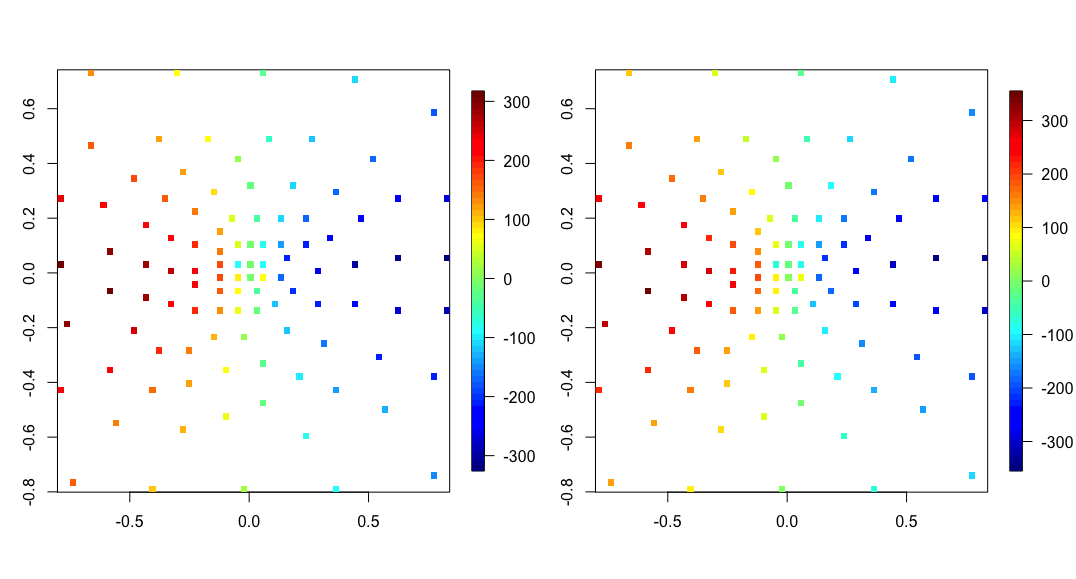}
         \caption{Left: $v_{S}$; Right: $v_{F}$}
     \end{subfigure}
     \begin{subfigure}{0.48\textwidth}
         \centering
         \includegraphics[height = 3cm, width = 8cm]{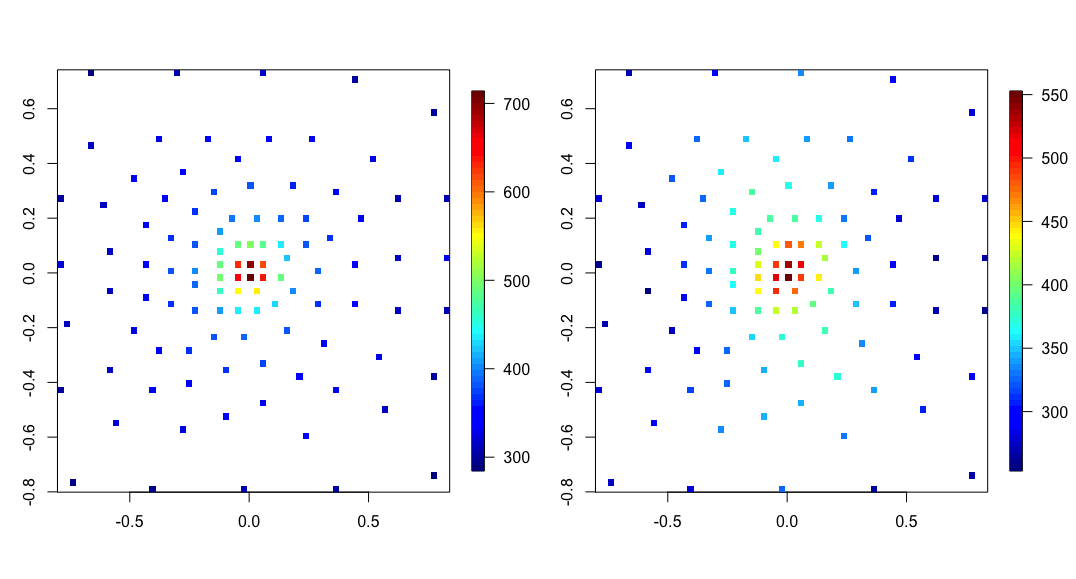}
         \caption{Left: $\tau_{S}$; Right: $\tau_{F}$}
     \end{subfigure}
     \hfill
     \begin{subfigure}{0.48\textwidth}
         \centering
         \includegraphics[height = 3cm, width = 8cm]{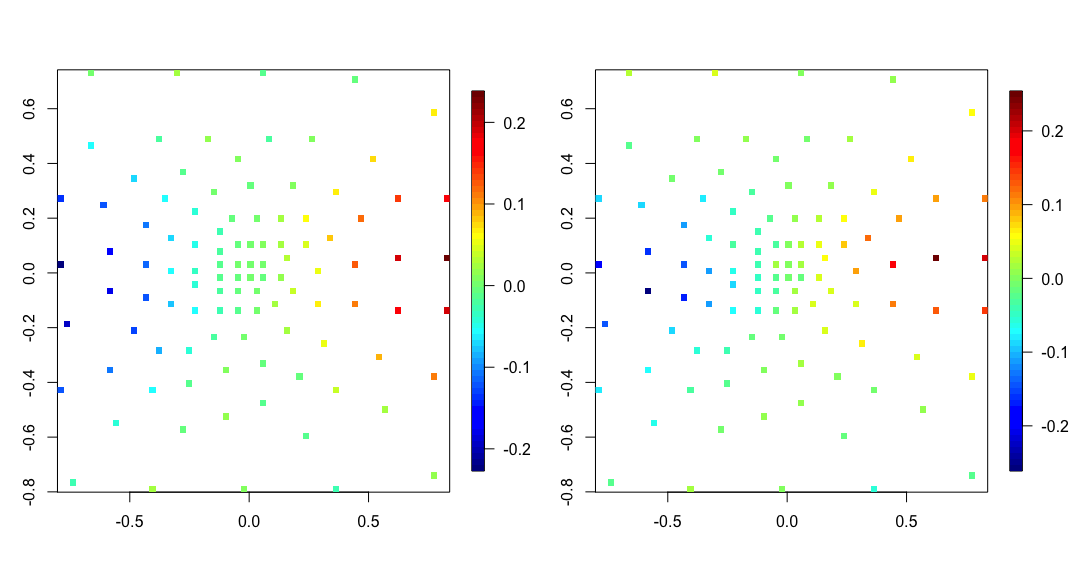}
         \caption{Left: $h_{3,S}$; Right: $h_{3,F}$}
     \end{subfigure}
     \begin{subfigure}{0.48\textwidth}
         \centering
         \includegraphics[height = 3cm, width = 8cm]{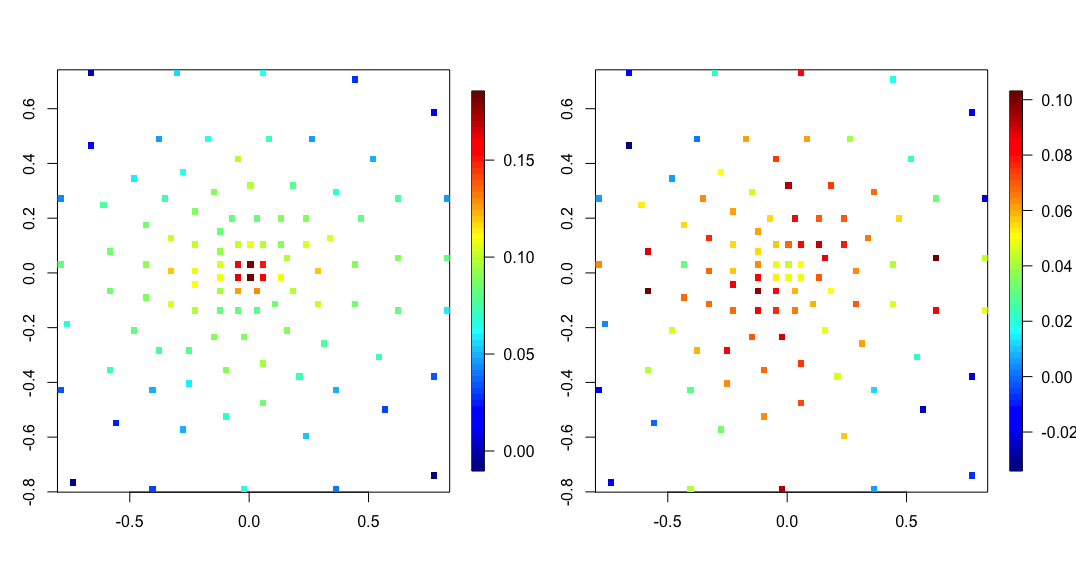}
         \caption{Left: $h_{4,S}$; Right: $h_{4,F}$}
     \end{subfigure}
     \caption{Orbital superposition output versus observed data on four moments of the velocity distribution for $\theta = (10.22, 10, 1)^\T$. The superposition method is implemented by a code originally developed by \cite{van2008triaxial}. }
     \label{fig:obsd_V_model}
\end{figure}

In it's current version \citep{van2008triaxial}, for a given input parameter combination $t = (t_1, t_2, t_3)^\T$, the code outputs the first four moments of the line-of-sight velocity distribution $f_S(x;t) = (v_{S}(x;t), \tau_{S}(x;t), h_{3,S}(x;t), h_{4,S}(x;t))^\T $ for points $x$ in a 2-dimensional spatial grid $\mathcal{X}$. The points $x$ are typically chosen to match the locations of the physical photometric data, providing further information about the mass distribution. Specifically, physical data $y_F(x_i) = (v_{F}(x_i), \tau_{F}(x_i), h_{3,F}(x_i), h_{4,F}(x_i))^\T$ is available for $i = 1, \ldots, n, \, n = 105$ locations of $\mathcal{X} = [-1, 1]^2$. In Figure \ref{fig:obsd_V_model}, color coded output from the code and the observed data are shown when $t = (10.22, 10, 1)^\T$ at the $n$ locations. This superposition technique is implemented by a computer code (FORTRAN) which is extremely computationally intensive; e.g. for one input of $\theta$, approximately 3 hours is needed to generate the output. The code output is available for three different values of $t_3 = 1, 2,3$, and for each value $t_3$, the output is available over a two-dimensional grid of values for $(t_1, t_2)$ where $t_1$ ranges from 10.22 to 10.31 and $t_2$ ranges from 10 to 10.25. Overall, the computer code output is available for $N = 476$ different values of $t$. 
Naturally, an exhaustive search over the input parameter space that compares the code output to the observed data to find the best-fit parameter combination is prohibitive. However, from limited experiments performed separately on each of the 4 outputs of the model, scientists have seen that for the point  $\tilde{\theta} = (9.92, 9.29, 1.27)$ the model outputs best approximated the observed data up to slight variations for each of the outcome. 
Here, we carry out a joint calibration of the Schwarzschild model. For this, we first estimate the covariance matrix as $\Sigma = \text{Cov}(E)$, where $E$ is a $n \times q$ error matrix obtained by fitting univariate BART regressions to the field data of each outcome. Here we fit the PBART method with functional projection. The marginal posterior distribution of $\theta$ is shown in Figure \ref{fig:multivariate_sc_theta}. The 95\% symmetric posterior (marginal) credible intervals are [9.54, 9.87], [8.78, 9.13] and [0.97, 1.38] for $\theta_1, \theta_2, \theta_3$, respectively. Intervals for $\theta_1, \theta_2$ from the joint analysis clearly seem to lie at the intersection of the intervals from the univariate model, but the posterior of $\theta_3$ exhibits slight bimodality. Overall, by modeling the four outcomes simultaneously uncertainty in $\theta$ have reduced which may significantly cut down computation time for future applications with an expanded parameter space. 
We also show the marginal posterior predictive mean of the outcomes.
In the left panel of Figure \ref{fig:calibrated_v}, we plot the posterior predictive mean at the observed locations for the four outcomes $v_R(\cdot), \tau_R(\cdot), h_{3,R}(\cdot), h_{4,R}(\cdot)$. For reference, the observed data is plotted on the right panel of Figure \ref{fig:calibrated_v}. As mentioned earlier, the outcomes in this context represent the first four moments of velocity, and the model performance decreases as the moments increase. We also observed a similar phenomenon when we looked at the relative squared error loss for a held-out data set of size 20. For $v_F(x)$,  We define the relative squared error loss as $n_t^{-1} \sum_{i = 1}^{n_t} (1 - \hat{v}_F(x_i)/v_F(x_i))^2$ where $\hat{v_F}(x_i)$ represents the posterior mean of samples of $[v_{S}(x_i, \theta) + b_{v, \theta}(x_i)]$. Here $n_t$ is the test data size. The definition is similar for other outcomes. The values we obtained are 0.17, 0.003, 3.92 and 30.42 for $v, \tau, h_3, h_4$, respectively.
\begin{figure}
    \includegraphics[width = \textwidth = 0.9, height = 4.5cm]{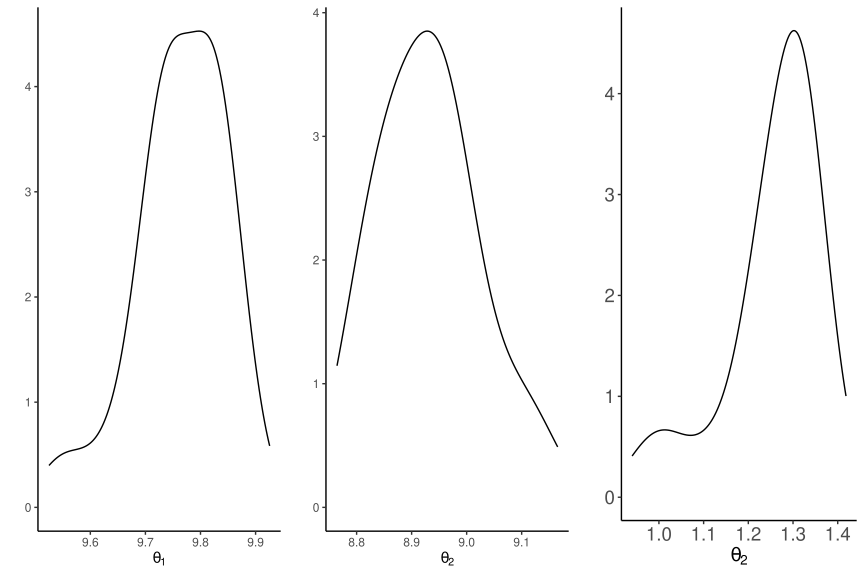}
    \caption{Posterior distribution of $\theta$ from multivariate PBART fit for the Schwarzschild model. }
    \label{fig:multivariate_sc_theta}
\end{figure}

\begin{figure}
     \centering
     \begin{subfigure}{0.4\textwidth}
         \centering
         \includegraphics[height = 4cm, width = 5cm]{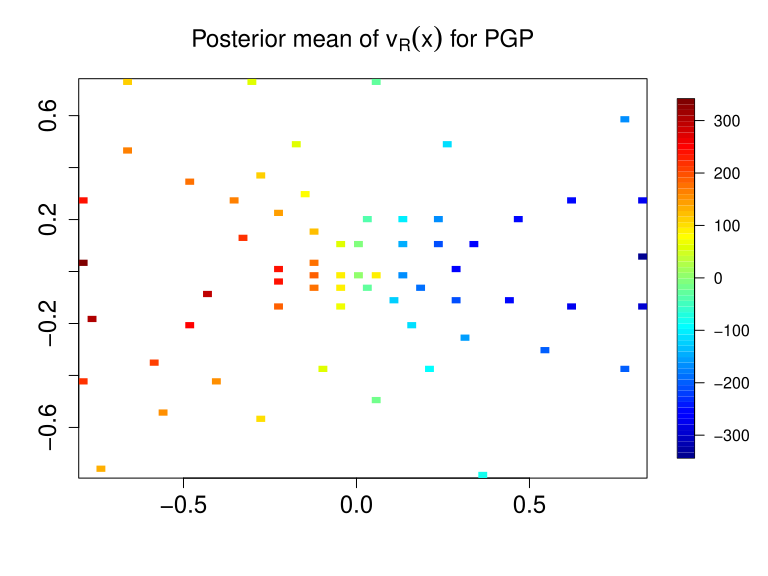}
     \end{subfigure}
     \begin{subfigure}{0.4\textwidth}
         \centering
         \includegraphics[height = 4cm, width = 5cm]{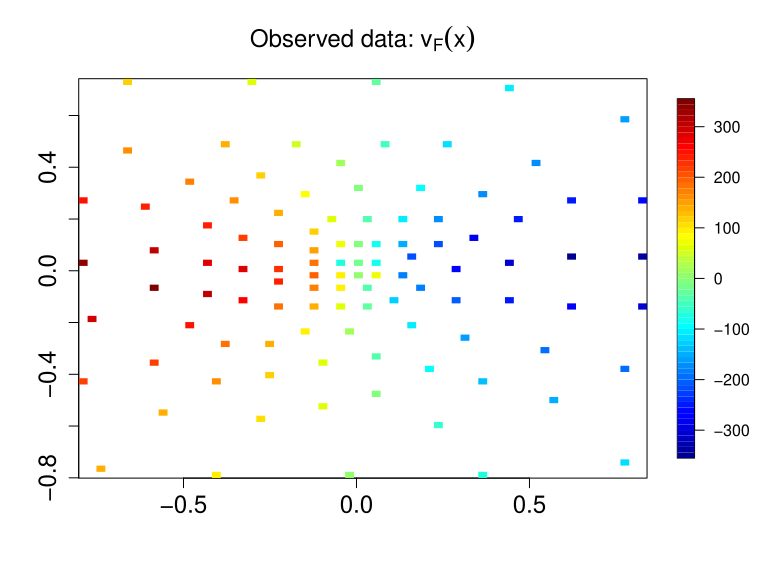}
     \end{subfigure}
     \begin{subfigure}{0.4\textwidth}
         \centering
         \includegraphics[height = 4cm, width = 5cm]{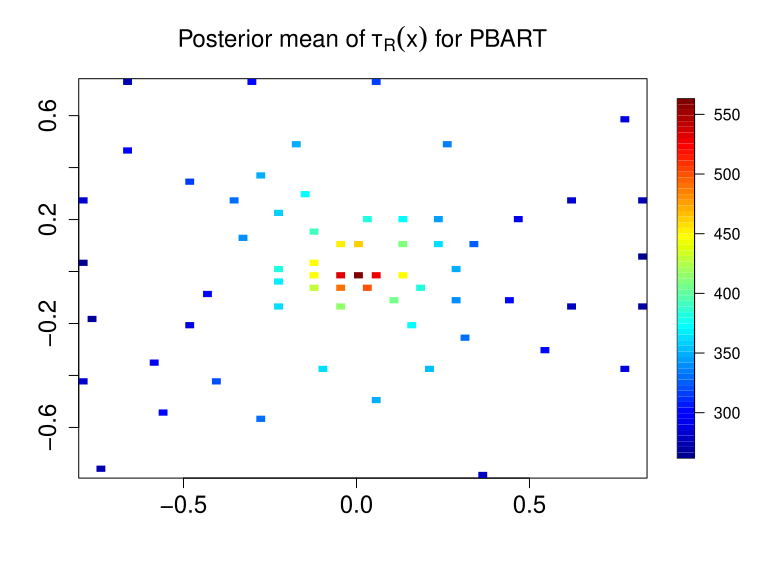}
     \end{subfigure}
     \begin{subfigure}{0.4\textwidth}
         \centering
         \includegraphics[height = 4cm, width = 5cm]{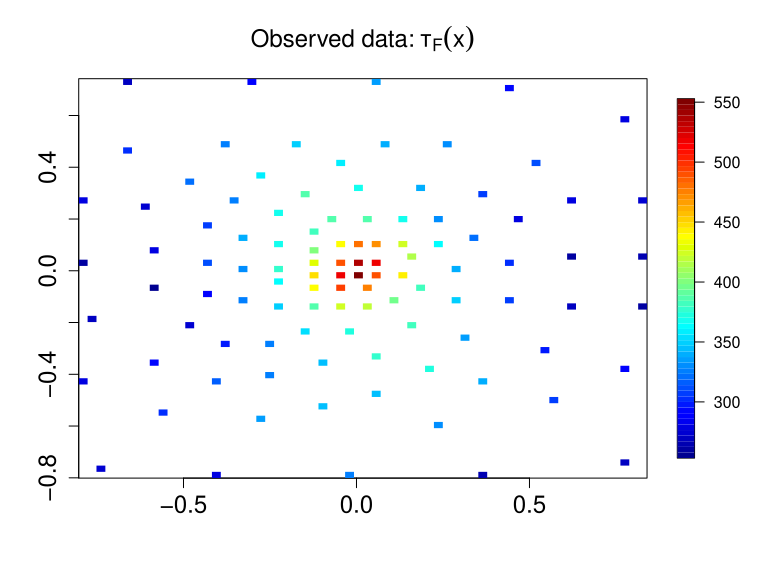}
     \end{subfigure}
     \begin{subfigure}{0.4\textwidth}
         \centering
         \includegraphics[height = 4cm, width = 5cm]{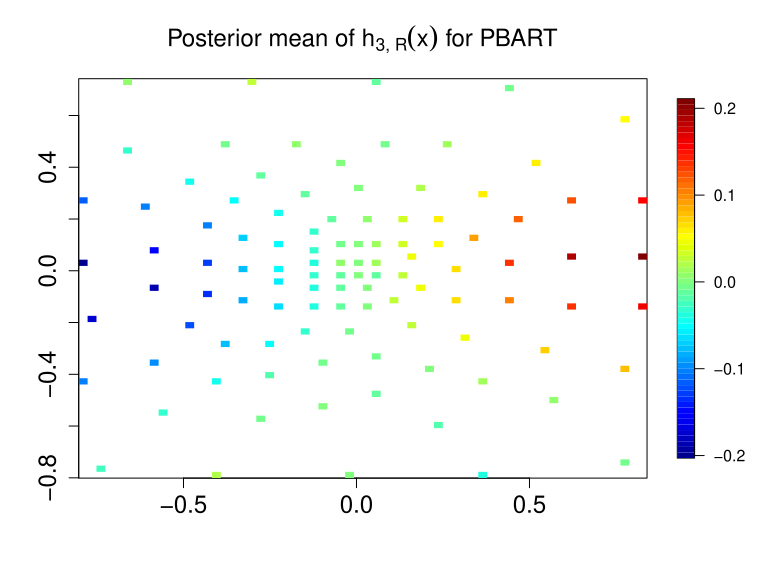}
     \end{subfigure}
     \begin{subfigure}{0.4\textwidth}
         \centering
         \includegraphics[height = 4cm, width = 5cm]{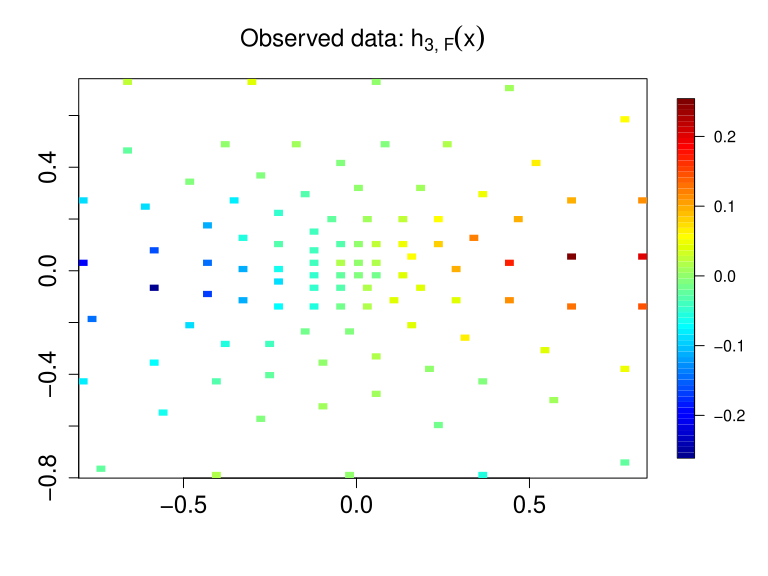}
     \end{subfigure}
     \begin{subfigure}{0.4\textwidth}
         \centering
         \includegraphics[height = 4 cm, width = 5cm]{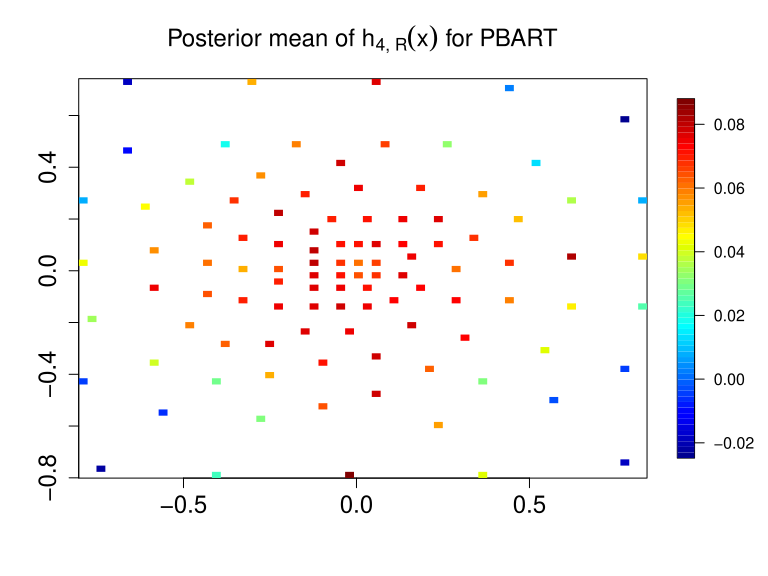}
     \end{subfigure}
     \begin{subfigure}{0.4\textwidth}
         \centering
         \includegraphics[height = 4cm, width = 5cm]{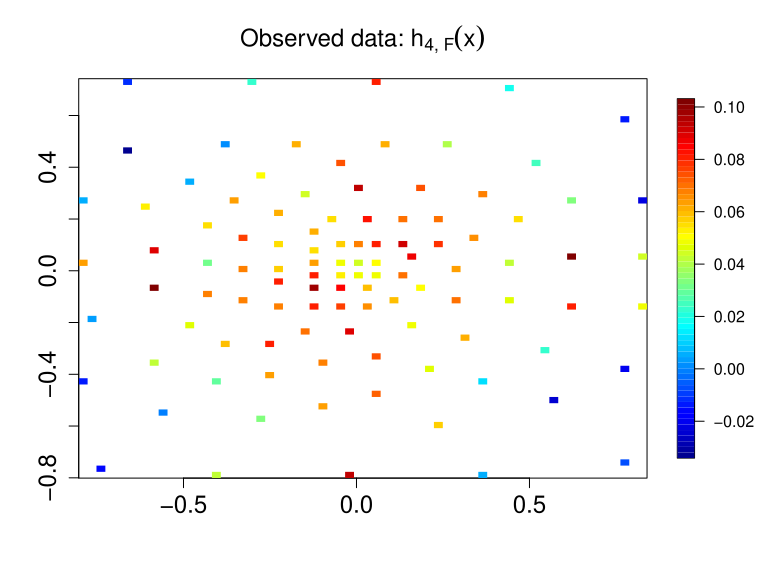}
     \end{subfigure}
     \caption{Posterior predictive mean of PBART is plotted for $v_R(\cdot), \tau_R(\cdot), h_{3,R}(\cdot), h_{4,R}(\cdot)$ in the left panel. In the second panel, the observed values of the corresponding quantities are also shown for reference. }
     \label{fig:calibrated_v}
\end{figure}

\section{Discussion}
The Schwarzschild model, although computationally intensive, is an important tool in understanding the dynamics of a black hole and its host galaxy. Motivated by this application we developed a multivariate calibration method that ensures parameter identifiablity under the squared error loss. The key benefits of the proposed projection posterior approach is its flexibility in accommodating user-specified prior distributions on the multivariate bias function, and the fact that such a projection is available analytically. Benefits of a multivariate analysis is demonstrated through numerical experiments and it seems to impact the analysis of Schwarzschild model positively as well. Conceptually, the projection approach can be extended to alternative loss functions that lead to different stationarity conditions and high-dimensional outcomes where low-rank models for the error covariance might be warranted.

\newpage
\begin{center}
{\large\bf SUPPLEMENTARY MATERIAL}
\end{center}
\medskip
(a) \emph{Supplementary Text:} contains proofs, and additional results from simulations and data analysis.\\
(b) \emph{Supplementary Code:} An \texttt{R} package to implement the proposed method can be downloaded from \href{https://github.com/antik015/OCal}{https://github.com/antik015/OCal}. Instructions for running the package and the code to reproduce Figure 1(a) - 1(c) and 2(a), 2(b) is also provided.





\bibliography{refs}
\bibliographystyle{chicago}

\setstretch{1.25}

\doublespacing
\clearpage
\begin{center}
	{\LARGE{\bf Supplementary Material to\\
	{\it Orthogonal calibration via posterior
projections with applications to the
Schwarzschild model}}
	}
\end{center}
\setcounter{equation}{0}
\setcounter{page}{1}
\setcounter{table}{0}
\setcounter{section}{0}
\setcounter{subsection}{0}
\setcounter{figure}{0}
\renewcommand{\theequation}{S.\arabic{equation}}
\renewcommand{\thesection}{S.\arabic{section}}
\renewcommand{\thealgorithm}{S.\arabic{algorithm}}
\renewcommand{\thepage}{S.\arabic{page}}
\renewcommand{\thetable}{S.\arabic{table}}
\renewcommand{\thefigure}{S.\arabic{figure}}


\section{Proofs}
\subsection{Proof of Proposition \ref{prop:multivariate_conditions}}
\begin{proof}
     Write $y_{R,k}(x) = f_{k, \theta}(x) + b_{\theta, k}(x)$ to obtain that the $j$-th element of $$\frac{\partial}{\partial t} L\{y_R(\cdot), f(\cdot, t)\}\rvert_{t =\theta} = \\ -2 \int_{\mathcal{X}} \sum_{k=1}^q  g_{j,k}(x) b_{\theta, k}(x) dx,$$ and recall $\theta^* = \argmin_{t \in \Theta} L\{y_R(\cdot), f(\cdot, t)\}$ which completes the proof.
\end{proof}
\subsection{Proof of Proposition \ref{prop:data_est}}
\begin{proof}
We first characterize $\theta$ in terms of the expected loss. Recall that for $Q = \mathrm{I}_q$, $\theta = \argmin_{t \in \Theta}  \int_x \sum_{k=1}^q \{y_{R,k}(x) - f_k(x, t)\}^2 dx$. For any fixed $t \in \Theta$, let $e(x) = \{y^{F}(x) - f(x, t)\}$, then
\begin{align*}
    &E_{P}\{e(x)^\T e(x)\}   = (\text{Vol}(\mathcal{X}))^{-1}\int_x \left[\sum_{k=1}^ q E_{P_{y\mid x}}\{e^2_k(x)\} \right]dx \\
    & = (\text{Vol}(\mathcal{X}))^{-1} \int_x \sum_{k=1}^q  E_{P_{y\mid x}}\{y_{F, k} - y_{R,k}(x)\}^2 dx + (\text{Vol}(\mathcal{X}))^{-1} \int_x \sum_{k=1}^q \{y_{R,k}(x) - f_k(x, t)\}^2 dx \\
    & = \text{tr}(\Sigma_F)  + (\text{Vol}(\mathcal{X}))^{-1} \int_x e(x)^\T e(x) dx,
\end{align*}
where $\text{tr}(A)$ denotes the trace of a matrix $A$.
Hence, $\theta$ can be seen as 
$$\argmin_{t \in \Theta} \int_{x} \sum_{k=1}^q\{y_{R,k}(x) - f_k(x, t)\}^2 dx = \argmin_{t \in \Theta}E_{P}[e(x)^\T e(x)].$$ 
Now define $Q_{n,k}(t) = n^{-1}\sum_{i=1}^n \{y_{F,k}(x_i) - f_k(x_i, t)\}^2$ and $Q_n(t) = q^{-1} \sum_{k=1}^q Q_{n,k}(t)$ which is simply $E_{P_n} [e(x)^\T e(x)]$ where $P_n$ is the corresponding joint empirical measure. Invoking Theorem 1 of \cite{pollard2006nonlinear} and the compactness of $\Theta$, 
we get that $\sup_f |Q_n(t) - E_P(Q_n(t))| \overset{P}{\to} 0$ which implies that $\theta_n^* \overset{P}{\to} \theta$. 
\end{proof}
\subsection{Proof of Lemma \ref{lm:projection_formula}}
\begin{proof}
The projection operator is the solution to the optimization problem 
\begin{align*}
    \min \norm{b_{\theta} - b_{\theta}^*}_{L^2_q} \quad \text{subject to } \sum_{k=1}^q \langle g_{j,k}, b_{\theta, k}^* \rangle = 0, \, j = 1, \ldots, p.
\end{align*}
Hence, the Lagrangian dual of the optimization problem is 
\begin{align*}
    \min \sum_{k=1}^q \norm{b_{\theta, k} - b_{\theta,k}^*}_L^2 + \sum_{j = 1}^p \lambda_j \left[\sum_{k=1}^q \langle g_{j,k}, b_{\theta, k}^* \rangle\right].
\end{align*}
Solving for $b_{\theta,k}^*$ gives $b_{\theta,k}^* = b_{\theta, k} +\sum_{j=1}^p \lambda_j g_{j,k}$. Thus, $\norm{b_{\theta, k} - b_{\theta,k}^*}_{L^2_q} = \lambda^\T (\sum_{k=1}^q Q_k) \lambda$ where $\lambda = (\lambda_1, \ldots, \lambda_p)^\T$ and $Q_k^{p \times p}$ is a matrix with $(j,j')$-th element $\langle g_{j,k} , g_{j',k}\rangle$. Also, $\sum_{k=1}^q \langle g_{j,k}, b_{\theta, k}^* \rangle = \sum_{k=1}^q \langle g_{j,k}, b_{\theta, k}\rangle + \sum_{j'=1}^p \lambda_{j'} Q_{k, jj'}$. {\color{black} Given the constraint that $\sum_{k=1}^q \langle g_{j,k}, b_{\theta, k}^* \rangle = 0$, we then have
$$\sum_{k=1}^q \langle g_{j,k}, b_{\theta, k}\rangle + \sum_{j'=1}^p \lambda_{j'} Q_{k, jj'}= 0 \iff Q\lambda + \eta = 0,$$
where $\eta \in \mathbb{R}^{p \times 1}$ with $\eta_j = \sum_{k=1}^q \langle g_{j,k}, b_{\theta, k} \rangle$. Hence, if $\lambda $ solves $Q\lambda = \eta$, we get $b_{\theta,k}^* = b_{\theta, k} -\sum_{j=1}^p \lambda_j g_{j,k}$.
}
\end{proof}
\section{Finite-dimensional projection}
When the central focus of analysis is assessing uncertainty in $\theta$, and $\Pi(b_k)$ is a Gaussian process, then the projection posterior approach can be implemented within a finite dimensional setting by ensuring $\sum_{k=1}^q \mathbf{b}_k^\T \mathbf{g}_{j,k} = 0$ where $\mathbf{b}_k = (b_k(x_1), \ldots, b_k(x_n))^\T$ and $\mathbf{g}_{j,k} = (g_{j,k}(x_1), \ldots, g_{j,k}(x_n))$. This is a consequence of posteriors of $\mathbf{b}_k$ being Gaussian under a Gaussian process prior and projections of multivariate Gaussian random vectors to linear subspaces. { \color{black} For example, suppose $X^{d \times 1} \sim \Gauss(\mu, \Sigma)$ and consider the set $S = \{x \in \mathbb{R}^d : A^\T x = 0\} \subset \mathbb{R}^d$ where $A^{d \times p}$ is matrix of rank $p$. Suppose $Y = A^\T X \in \mathbb{R}^p$, so that $\text{Cov}(X, Y) = \Sigma A$. Then, the joint distribution of $Z = (X, Y)$ is a $(d+p)$-dimensional Gaussian with parameters
\begin{align*}
    \mu_Z = \begin{pmatrix}
        \mu \\
        A^\T \mu
    \end{pmatrix}, \quad \Sigma_Z = \begin{pmatrix}
        \Sigma & \Sigma A \\
        A^\T \Sigma & A^\T \Sigma A 
    \end{pmatrix}.
\end{align*}
Moreover, $X \mid Y = y \sim \Gauss( \mu + \Sigma A (A^\T \Sigma A)^{-1} (y - A^\T \mu), \Sigma - \Sigma A (A^\T \Sigma A)^{-1} A^\T \Sigma)$. In particular, when $y = 0$, we obtain $X \mid X \in S  \sim \Gauss(\mu - \Sigma A (A^\T \Sigma A)^{-1}A^\T \mu, \Sigma - \Sigma A (A^\T \Sigma A)^{-1}A^\T \Sigma)$. Next, define the following projection for any $x \in \mathbb{R}^d$
\begin{equation}\label{eq:whitened_projection}
    P_S(x) = \Sigma^{1/2} \argmin_{y\in S} \norm{y - \Sigma^{-1/2}x}.
\end{equation}
Standard linear algebra yields that $P_S(x)$ has the form $\Sigma^{1/2}(I - P_A)\Sigma^{-1/2} x$ where $P_A = A(A^\T A)^{-1} A^\T$ is the projection matrix of $A$. This leads to the following interpretation of the multivariate Gaussian distribution conditional on linear equality constraints. 
\begin{proposition}\label{prop:projection_equivalence}
Suppose $X \sim \Gauss(\mu, \Sigma)$ and consider the set $S = \{x: A^\T x = 0\} \subset \mathbb{R}^d$ where $A^{d \times p}$ is matrix of rank $p$. Then the random variable $$P_S(X) = \Sigma^{1/2}(I - P_A)\Sigma^{-1/2}X \sim \Gauss(\mu - \Sigma A (A^\T \Sigma A)^{-1} A^\T \mu,\, \Sigma - \Sigma A (A^\T \Sigma A)^{-1}A^\T \Sigma).$$
\end{proposition}
\begin{proof} 
The covariance of $P_S(X) = \Sigma^{1/2}(I - P_A)^2 \Sigma^{1/2}$. Next, we have
\vspace{-0.2in}
\begin{align*}
    \text{Cov}(X \mid X \in S) &= \Sigma - \Sigma A (A^\T \Sigma A)^{-1}A^\T \Sigma \\
    &= \Sigma ^{1/2} (I - \Sigma^{1/2}A (A^\T \Sigma A)^{-1}A^\T \Sigma^{1/2})\Sigma^{1/2}\\
    & = \Sigma^{1/2} (I - P_B) \Sigma^{1/2} = \Sigma^{1/2} (I - P_B)^2 \Sigma^{1/2} = \text{Cov}(P_S(X)),
\end{align*}
where $B = \Sigma^{1/2} A$ and $P_B = B (B^\T B)^{-1} B^\T$. Also, $P_B = P_A$ since $A$ is full rank and $\Sigma$ is non-singular by assumption. A similar analysis shows that the mean also agrees.
\end{proof}
}
Next, we illustrate the consequence of this interpretation within the orthogonal calibration context. For the sake of simplicity, consider $q=1$. We set $A^{n \times p} = (a_1, \ldots, a_p) $ where $a_j = (g_j(x_1, \tilde{\theta}), \ldots, g_j(x_n, \tilde{\theta}))^\T$. Let $b = (b(x_1), \ldots, b(x_n))^\T$ denote the vector of bias function evaluated at $x_1, \ldots, x_n$. Then the finite-dimensional analogue of the orthogonality condition is $A^\T b = 0$. When $b_{\tilde{\theta}}(\cdot)$ is endowed with a zero mean Gaussian process prior with covariance kernel $C(\cdot, \cdot)$, then {\it apriori} the vector $b_{\tilde{\theta}} \sim \Gauss(0, K)$ with $K_{ij} = C(x_i, x_j),\, 1\leq i, j\leq n$. Also, $b_{\tilde{\theta}} \mid \theta, y^{(n)} \sim \Gauss( (K + \sigma^2_F)^{-1} z^{(n)}, (K + \sigma^2_F)^{-1} )$ where $z^{(n)} = (y(x_1) - f(x_1, \theta), \ldots, y(x_2) - f(x_n, \theta))^\T$. Using Proposition \ref{prop:projection_equivalence}, we can then obtain $b^*$. Hence, for Gaussian process priors, only the second step in Algorithm 1 from the main draft will change under the finite-dimensional projection setup. 

We now demonstrate how the finite-dimensional projection can be executed when $b_{\tilde{\theta}}$ is specified a non-Gaussian process prior. Here the posterior of $\mathbf{b}_k$ is also non-Gaussian. Suppose the prior is parameterized by $\eta$. For instance, with a BART prior, $\eta$ contains the trees and leaf node parameters. A prior on $b_{\tilde{\theta}}$ is specified by a prior distribution on $\eta$ which then produces a full conditional distribution of $\eta \mid \theta, y^{(n)}$. A posterior sample of $b_{\tilde{\theta}} \in \mathbb{R}^n$ can be drawn by sampling from $\eta \mid \theta, y^{(n)}$. For many priors, the joint distribution can be well approximated by a multivariate Gaussian with some mean and covariance $(\beta, \Phi)$ \citep{yang2017frequentist}. To find the approximating mean and covariance, we simply draw $M$ independent draws of $b_{\tilde{\theta}}$ by sampling $\{\eta_1, \ldots, \eta_M\}$ independently from $\eta\mid \theta, y^{(n)}$ and set $\beta = M^{-1} \sum_{m=1}^M b_{\tilde{\theta}, m} $ and $\phi = (M-1)^{-1} (b_{\tilde{\theta}, m} - \beta)(b_{\tilde{\theta}, m} - \beta)^\T$ where $b_{\tilde{\theta}, m}$ is the draw of $b_{\tilde{\theta}}$ corresponding to $\eta_m,\, m = 1, \ldots, M$. We note here, that $M$ can be made arbitrarily large compared to the dimension $n$ of $b_{\tilde{\theta}}$ so that the sample mean and covariance provide reasonable approximations to their population counterparts. This technique is summarized in the following algorithm.

\begin{algorithm}
\caption{Projection sampler 2: non-GP priors}
1. Sample $M$ independent copies of parameters $\eta \mid \theta, y^{(n)}$ and consider the corresponding draws of $b_{\tilde{\theta}}$ as $\{b_{\tilde{\theta}, 1}, \ldots, b_{\tilde{\theta, M}}\}$.

2. Compute the quantities $\beta = M^{-1} \sum_{m=1}^M b_{\tilde{\theta}, m} $ and $\Phi = (M-1)^{-1} (b_{\tilde{\theta}, m} - \beta)(b_{\tilde{\theta}, m} - \beta)^\T$.

2. Set the projection as $b_{\tilde{\theta}}^* = \Phi^{1/2}(I - P_G)\Phi^{-1/2}b_{\tilde{\theta}}$.

3. Update $\theta \sim \Pi(\theta \mid b^*_{\tilde{\theta}}, y^{(n)})$.

\label{algo:conditional_sampler2}
\end{algorithm}
Clearly, for non-Gaussian priors for $b_k$, the finite-dimensional projection approach is computationally much more expensive. However, if the prior is a GP, this approach is relatively less expensive to implement and yields similar results to the functional projection approach described in the main draft as can be seen from the results of our experiments in Tables \ref{tab:case1_finite} and \ref{tab:case2_finite}.

\begin{table}
    \centering
    \scalebox{0.9}{
    \begin{tabular}{ccc}
    \hline
         \multicolumn{3}{c}{\textbf{Model 1 ($\theta^* = 3.56$)}} \\
         \hline
         & PGP(Fn) & PBART(Fn)   \\
         \hline
         Mean & 3.59 & 3.58  \\
        Std. Dev. & 0.07 & 0.06 \\
         Coverage & 0.91 & 0.90 \\
         Runtime &30.72 sec & 17 min \\
         \hline
    \end{tabular}
    }
    \caption{Simulation results for the finite-projection method for \textbf{Model 1}. We report the posterior mean, posterior standard deviation, coverage of the 95\% credible intervals and the average runtime.}
    \label{tab:case1_finite}
\end{table}

\begin{table}
    \centering
    \scalebox{0.7}{
    \begin{tabular}{ccc}
    \hline
          \multicolumn{3}{c}{\textbf{Model 2 ($\theta^* = (0.2, 0.3)$)}} \\
         \hline
         & PGP(Fn) & PBART(Fn)\\
         \hline
         Mean & (0.2, 0.31) & (0.2, 0.31) \\
        Std. Dev. & (0.0003, 0.0004)  & (0.0008, 0.0008)\\
         Coverage & 0.91 & 0.92 \\
         Runtime & 2.65 min & 1.35 hrs \\
         \hline
    \end{tabular}
    }
    \caption{Simulation results for the finite-projection method for \textbf{Model 2}. We report the posterior mean, posterior variance, coverage of the 95\% credible intervals and the average runtime.}
    \label{tab:case2_finite}
\end{table}

\section{Univariate analysis on Schwarzschild model outputs}
We apply the proposed posterior projection technique on each output of the Schwarzschild model separately. We illustrate the specifics of our application with the first output of the Schwarzschild model, which is the mean velocity. Similar to the \cite{kennedy2001bayesian} setup, we have the following model for the observed velocities $v_F(x)$ and the corresponding code output $v_S(x)$  at location $x \in \mathcal{X}$:
\begin{equation}
    v_F(x_i) = v_R(x_i) + \epsilon_i = v_S(x_i;\tilde{\theta}) + b_{v,\tilde{\theta}}(x_i) + \epsilon_i, \quad \epsilon_i \sim \Gauss(0, \sigma_v^2),
\end{equation}
where we let $\tilde{\theta}$ to be the parameter value where the $L^2$ loss defined in previous sections is minimized and $b_{v, \tilde{\theta}}(\cdot)$ is the bias function corresponding to mean velocity at $\tilde{\theta}$. Since the explicit form of the model is not available, we use the posterior mean of a BART fit as the definition of $v_S(x;t)$ based on the model outputs only. We fix $\sigma_v$ to the posterior mean of $\sigma$ of a BART regression fit to $(v_{F,i}, x_i)_{i=1}^n$. We set $\theta_j \sim \Gauss(0, \gamma^2)$ with $\gamma = 10$ as priors for $\theta_j$, and $\Pi(\theta) = \prod_{j=1}^p \Pi(\theta_j)$. We implement PGP, PBART and OGP (with a Matern kernel having smoothness parameter 5/2) under these settings where for implementing Step 3 of Algorithm \ref{algo:conditional_sampler}, we use an adaptive Metropolis-Hastings sampler \cite{haario2001adaptive} which is tuned to maintain an acceptance rate of approximately 0.3. 

In Table \ref{tab:theta_sch}, we summarize the posterior distribution of the calibration parameters for all the four different outputs of the Schwarzschild model. Specifically, we report the posterior mean, standard deviation and the 95\% credible interval for all the three calibration parameters $(\theta_1, \theta_2, \theta_3)$. As expected, there are some differences in the posterior means for different velocity moments being fit. For example, the posterior mean for $\theta_1$ differs when using $h_4$ compared to when using $\tau$. The difference may be due to the physical degeneracy between the orbital anisotropy and black hole mass, such that large $\tau$ at the center of the 2-dimensional spatial grid can imply a large black hole mass but so too can a smaller $\tau$ with radial anisotropy, which is reflected in the values of $h_4$. Or, it could be due simply to the model inadequacy of a univariate treatment of the problem. In addition, there seems to a broad agreement between the posterior means for $\theta_1$ and $\theta_2$ under the two priors. However, the key difference is in the posterior mean of $\theta_3$. For the GP prior, the posterior mean for the 4 outcomes is roughly around 1, whereas for BART, the posterior mean is around 1.4. One possible explanation for this discrepancy is that for $\theta_3$, the computer code output is only available for $\theta_3 = 1, 2, 3$. The scarcity of data along this dimension may potentially lead to unstable estimates. A joint model accounting for the covariance between the measurements can reconcile this apparent discrepancy in the posterior distribution and the problems induced by sparse data. Numerical experiments in the next section illustrates the benefits of a joint model when outcomes are correlated. 


\begin{table}
    \centering
    \scalebox{0.6}{
    \begin{tabular}{c|c|ccc|ccc|ccc}
    \hline
         &  & \multicolumn{3}{c}{PGP} & \multicolumn{3}{c}{PBART}& \multicolumn{3}{c}{OGP} \\
         \hline
         &  & $\theta_1$ & $\theta_2$ & $\theta_3$ & $\theta_1$ & $\theta_2$ & $\theta_3$ & $\theta_1$ & $\theta_2$ & $\theta_3$ \\
         \hline
         \multirow{3}{*}{$v$} &Mean & 9.46 & 9.06 & 0.82 & 9.59 & 9.16 & 1.22 & 9.51 & 9.08 & 1.10\\
         & Std. Dev. & 0.14 & 0.08 & 0.09 & 0.27 & 0.15 & 0.07 & 0.37 & 0.26 & 0.18\\
         & Interval & (9.21, 9.80) & (8.93, 9.18) & (0.65, 1.00)& (9.20, 9.97) & (8.93, 9.50) & (1.12, 1.32) &  (8.45, 10.88) & (8.17, 9.6) & (0.52 , 1.60)\\
         \hline
         \multirow{3}{*}{$\tau$} & Mean & 9.98 & 9.12 & 1.02 & 9.59 & 9.03 & 1.42 & 9.82 & 9.24 & 1.31 \\
         & Std. Dev. & 0.09 & 0.07 & 0.10 & 0.10 & 0.13 & 0.18 & 0.21 & 0.19 & 0.30 \\
         & Interval & (9.81, 10.14) & (8.99, 9.26) & (0.87, 1.23) & (9.41, 9.76) & (8.84, 9.31) & (1.15, 1.76) & (9.23, 10.49) & (8.71, 9.81) & (0.35, 2.13) \\
         \hline
         \multirow{3}{*}{$h_3$} & Mean & 9.96 & 9.36 & 1.04 & 10.09 & 9.27 & 1.41 &  9.98 & 9.31 & 1.21  \\
         & Std. Dev. & 0.07 & 0.08 & 0.12 & 0.07 & 0.08 & 0.06 & 0.22 & 0.17 & 0.24 \\
         & Interval & (9.86, 10.11) & (9.15, 9.49) & (0.87, 1.25)& (9.97, 10.23) & (9.12, 9.41) & (1.25 , 1.41) & (9.36 , 10.69) & (8.84, 9.82) & (0.47, 1.85)  \\ 
         \hline 
         \multirow{3}{*}{$h_4$} & Mean & 10.18 & 9.07 & 1.09 & 9.73 & 9.09 & 1.53 & 10.16 & 9.11 & 1.42\\
         & Std. Dev. & 0.11 & 0.06 & 0.06 & 0.06 & 0.12 & 0.09 & 0.19 & 0.21 & 0.23 \\
         & Interval & (9.89, 10.62) & (8.97, 9.20) & (0.97, 1.19) & (9.61, 9.82) & (8.85, 9.28) & (1.35, 1.69) & (9.63, 10.77) & (8.53, 9.74) & (0.71, 2.03)\\
         \hline
    \end{tabular}
    }
    \caption{Posterior summaries for the calibration parameters $\theta_1$ = logarithm of mass of the black hole, $\theta_2$ = stellar mass-to-light ratio and $\theta_3$ = fraction of dark matter. The posterior mean, standard deviation (Std. Dev.) and the 95\% credible interval is reported here.}
    \label{tab:theta_sch}
\end{table}
\section{Calibration with multivariate outcomes}
The numerical experiments in the main draft were mostly concerned when one outcome is available. Here, we consider a simulation scenario where more than one outcome is measured, i.e. $q > 1$. We illustrate this with the case $q = 2$. For $x \in [0,1]$, consider the model $y_{F,1}(x) = f_1(x, \theta) + b_{\theta, 1} (x) + \epsilon_1$, and $ y_{F,2}(x) = f_2(x, \theta) + b_{\theta, 2} (x) + \epsilon_2$
where $y_{R, 1}(x) = 4x + x\sin 5 x$, $y_{R,2}(x) = [1 + \exp\{-6(x - 0.5)\}]^{-1}$ are the real data-generating processes. The respective computer models are $f_1(x, t) = tx$, $f_2(x, t) = \Phi\{t(x - 0.5)\}$ with $\Phi(\cdot)$ being the cdf of $Z \sim \Gauss(0,1)$. The errors are assumed to follow $\Gauss(0, \Sigma)$. 
The loss $L\{y_R(\cdot), f(\cdot, t)\} = \sum_{k = 1}^2 \int_{\mathcal{X}} \{y_{R, k}(x) - f_k(x, t) \}^2$ is minimized at $ \theta^* = 3.56$ approximately. However, the behavior of the $L_2$-loss in these two cases is very different, especially near the minimum. Indeed, as shown in the left panel of Figure \ref{fig:loss_comparison}, for the pair $(y_{R,1}, f_1)$, the loss increases quite sharply around 3.56, whereas for $(y_{R,2}, f_2)$, the increase in loss for $t>3.56$ is much slower. As such, when data from this model is fitted separately in the univariate framework, one should expect more uncertainty in $\theta$ for $(y_{R,2}, f_2)$. On the other hand, a joint multivariate fit should reflect the fact when the two real processes and their corresponding computer models are compared together, the posterior distribution of $\theta$ is much more centered around 3.56. 
\begin{figure}
     \centering
     \begin{subfigure}{0.4\textwidth}
         \centering
         \includegraphics[height = 3 cm, width = 5cm]{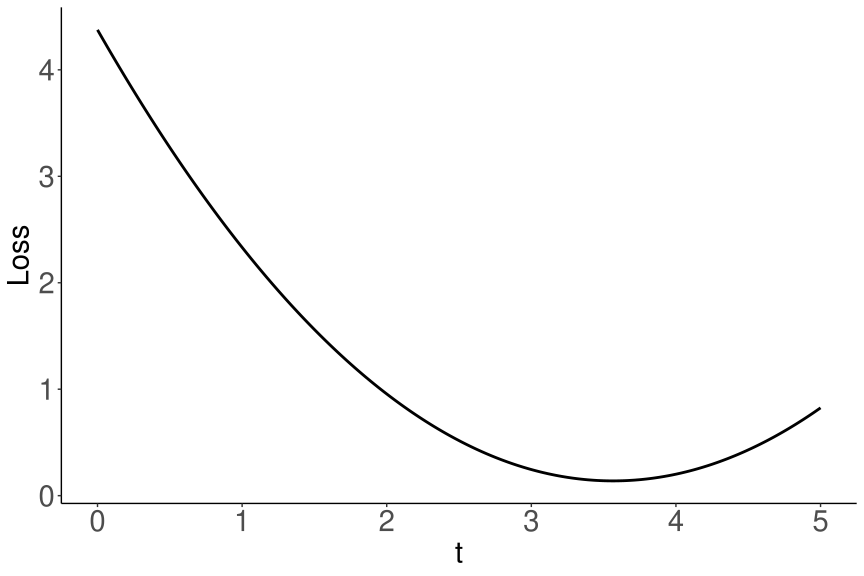}
         \caption{$L(t) = \int_0^1 \left\lbrace y_{R,1}(x) - f_1(x, t)\right\rbrace^2 dx$.}
     \end{subfigure}
     \begin{subfigure}{0.4\textwidth}
         \centering
         \includegraphics[height = 3cm, width = 5cm]{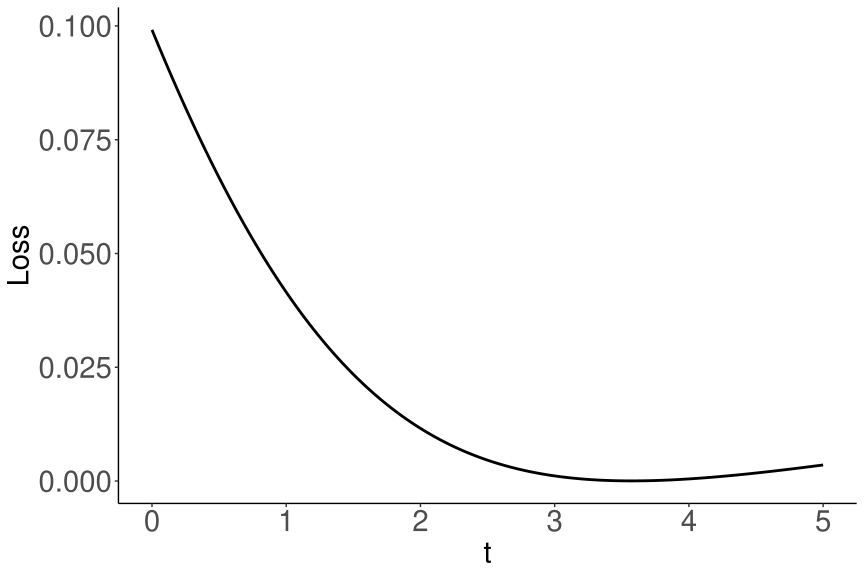}
         \caption{$L(t) = \int_0^1 \left\lbrace y_{R,2}(x) - f_2(x, t)\right\rbrace^2 dx$.}
     \end{subfigure}
     \caption{Loss comparison}
     \label{fig:loss_comparison}
\end{figure}

For data generation, we consider $n = 100$ and a covariance matrix $\Sigma_0$ such that the diagonal elements are $0.2^2$ and the off-diagonal elements are $0.012$. We consider PGP and PBART for this setup and use the functional projection. We fit this data separately in the univariate framework and then a joint model is fitted. For a fair comparison all hyperparameter values were the same. We plot the posterior distribution of $\theta$ obtained using the PBART method in Figure \ref{fig:univariate_multivariate}. In the first and second panel, the posterior distribution of $\theta$ is shown when two separate univariate models are fitted, and in the third panel we show that for a multivariate fit. As mentioned earlier, the posterior distribution of $\theta$ when data for the second outcome is fitted separately, is more dispersed and has relatively large mass in the interval [3, 4]. This is not the case for the first outcome, as almost the entire mass is within [3.5, 3.65]. However, the posterior distribution of $\theta$ form a multivariate fit is able to borrow information across the two outcomes so that most of the mass of the distribution is within [3.4,3.7].
\begin{figure}
    \includegraphics[width = \textwidth = 0.9, height = 5cm]{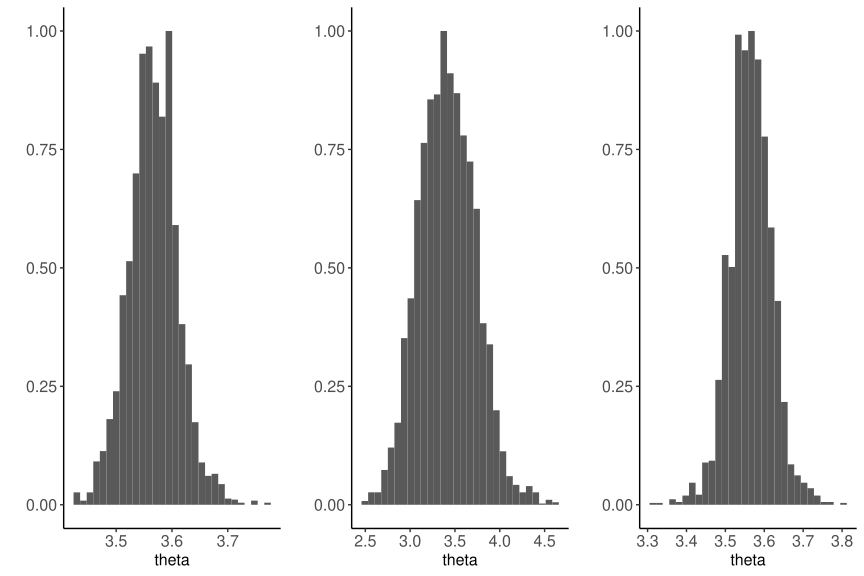}
    \caption{Posterior distribution of $\theta$ for univariate and multivariate PBART fits for the multivariate simulation model. Left and middle panel show posteriors for two separate univariate fits, and the right panel shows posterior from a combined multivariate fit.}
    \label{fig:univariate_multivariate}
\end{figure}

\end{document}